\documentclass[a4paper,onecolumn,11pt,accepted=2020-07-09]{quantumarticle}
\pdfoutput=1
\usepackage[utf8]{inputenc}
\usepackage[english]{babel}
\usepackage[T1]{fontenc}
\usepackage{amsmath}
\usepackage{hyperref}

\usepackage{color}
\usepackage{palatino}
\usepackage{bm}
\usepackage{amsfonts}
\usepackage{amscd}
\usepackage{amsthm}
\usepackage{amssymb}
\usepackage{mathtools}
\usepackage{graphics, graphicx}
\usepackage{verbatim}
\usepackage{dsfont}
\usepackage{enumitem}
\usepackage{setspace}
\usepackage[numbers,sort,compress]{natbib}

\newtheorem{theorem}{Theorem}
\newtheorem{proposition}{Proposition}[section]
\newtheorem{lemma}{Lemma}[section]
\newtheorem{corollary}{Corollary}[section]

\newtheorem{definition}{Definition}[section]

\newcommand{\bra}[1]{\mbox{$\left\langle #1 \right|$}}
\newcommand{\ket}[1]{\mbox{$\left| #1 \right\rangle$}}

\newcommand{\op}[2]{|#1\rangle\langle #2|}

\newcommand{\ceil}[1]{\left\lceil #1 \right\rceil}
\newcommand{\tens}{\mathbin{\mathop{\otimes}}}
\newcommand{\Tr}{\operatorname{Tr}}
\newcommand{\tr}{\operatorname{Tr}}
\newcommand{\1}{\mathds{1}}
\newcommand{\mc}{\mathcal}
\newcommand{\mbb}{\mathbb}
\newcommand{\mr}{\mathrm}
\newcommand{\id}{\mathrm{id}}
\newcommand{\pr}{\mathrm{pre}}
\newcommand{\post}{\mathrm{post}}

\definecolor{cool_green}{rgb}{0.0, 0.5, 0.0}

\begin{document}

\title{Entanglement-breaking superchannels}
\author{Senrui Chen}
\email{csenrui@gmail.com}
\affiliation{Department of Electronic Engineering, Tsinghua University, Beijing 100084, China}
\affiliation{Center for Quantum Information, Institute for Interdisciplinary Information Sciences, Tsinghua University, Beijing 100084, China }
\author{Eric Chitambar}
\email{echitamb@illinois.edu}
\affiliation{Department of Electrical and Computer Engineering, Coordinated Science Laboratory, University of Illinois at Urbana-Champaign, Urbana, IL 61801}


\maketitle

\begin{abstract}
\begin{spacing}{1}

In this paper we initiate the study of entanglement-breaking (EB) superchannels.  These are processes that always yield separable maps when acting on one side of a bipartite completely positive (CP) map. EB superchannels are a generalization of the well-known EB channels. We give several equivalent characterizations of EB supermaps and superchannels. Unlike its channel counterpart, we find that not every EB superchannel can be implemented as a measure-and-prepare superchannel.  We also demonstrate that many EB superchannels can be superactivated, in the sense that they can output non-separable channels when wired in series. 

We then introduce the notions of CPTP- and CP-complete images of a superchannel, which capture deterministic and probabilistic channel convertibility, respectively.  This allows us to characterize the power of EB superchannels for generating CP maps in different scenarios, and it reveals some fundamental differences between channels and superchannels.  Finally, we relax the definition of separable channels to include $(p,q)$-non-entangling channels, which are bipartite channels that cannot generate entanglement using $p$- and $q$-dimensional ancillary systems. By introducing and investigating $k$-EB maps, we construct examples of $(p,q)$-EB superchannels that are not fully entanglement breaking. Partial results on the characterization of $(p,q)$-EB superchannels are also provided.
%
%

\end{spacing}

\end{abstract}

\newpage

\tableofcontents

\newpage

\section{Introduction}

\sloppy Suppose that Alice and Rachel have access to some bipartite quantum channel $\mc{E}^{A_0R_0\to A_1R_1}$.  They are interested in using this channel to generate entangled states across their spatially separated laboratories.  As shown in Fig.~\ref{fig:entangling-state}, the most general method for doing so would involve using local quantum memories.  Alice prepares a locally entangled state $\rho^{A_0A_E}$, with $A_E$ being her memory register, and Rachel does likewise with the state $\omega^{R_0R_E}$.  Sending systems $A_0$ and $R_0$ through the channel leads to the state 
\[\sigma^{A_EA_1:R_ER_1}=\id^{A_ER_E}\otimes\mc{E}^{A_0R_0\to A_1R_1}\left(\rho^{A_0A_E}\otimes \omega^{B_0B_E}\right),\]
which they hope is entangled.  It is known that such a procedure can be used to generate entanglement if and only if $\mc{E}^{A_0R_0\to A_1R_1}$ does not have the form of a so-called \textit{separable} channel \cite{Cirac-2001a}.  Hence for Alice and Rachel's goal of obtaining bipartite entangled states, separable channels are completely useless.

\begin{figure}[b]
    \centering
    \includegraphics[width=0.55\columnwidth]{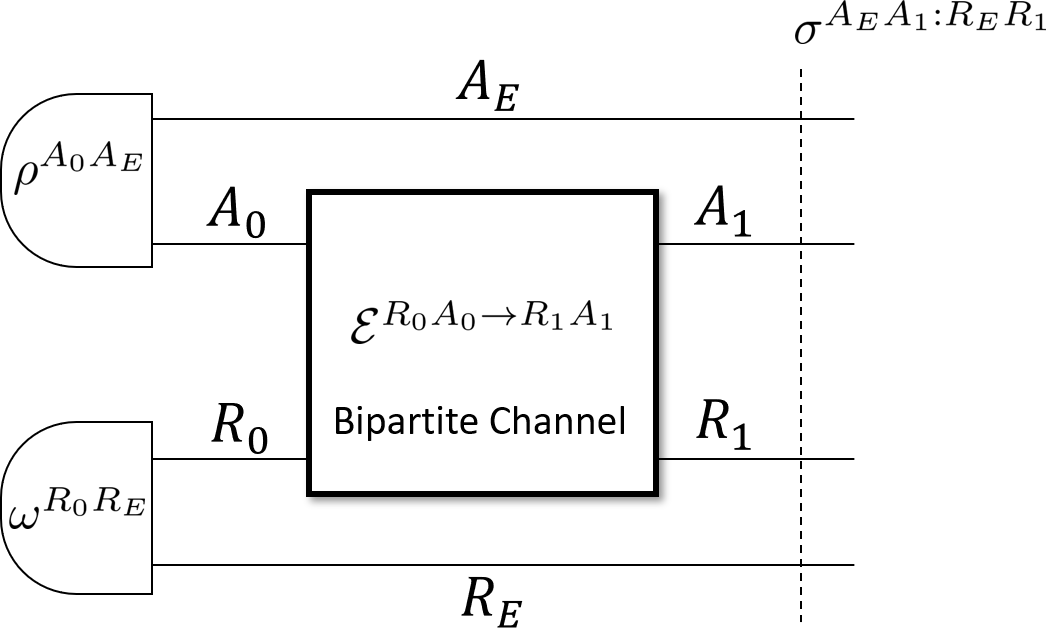}
    \caption{Alice and Rachel can use the their bipartite channel $\mc{E}$ with local quantum memories to generate entanglement in $\sigma^{A_EA_1:R_ER_1}$ if and only if $\mc{E}$ is not a separable channel.}
    \label{fig:entangling-state}
\end{figure}

Frustrated with the situation, Alice naively wonders if manipulating her part of the channel could improve their prospects of obtaining entanglement.  Any physical procedure she attempts can be described as in Fig.~\ref{fig:local-superchannel}; it involves her first applying some pre-processing map that couples her input system $A_0$ to the memory register $A_E$, and then applying a post-processing map to system $A_E$ and her channel output $A_1$ \cite{Chiribella2008}.  Such a process is known as a superchannel, and specifically here it is a local superchannel since it is being implemented only in Alice's laboratory.  Unfortunately for Alice, local superchannels are not able to transform a separable channel into a non-separable one.  Consequently, if $\mc{E}^{A_0R_0\to A_1R_1}$ is useless for entanglement generation before Alice's manipulation, it will be useless after.  On the other hand, it is quite possible that a local superchannel converts a non-separable channel into a separable one.  This begs the question of whether there exist certain local superchannels for Alice that convert \textit{every} bipartite channel into a separable channel.  We refer to such processes as \textit{entanglement-breaking superchannels} since they completely eliminate any channel's ability to distribute entanglement, and they are the focus of this paper.

\begin{figure}[t]
    \centering
    \includegraphics[width=0.65\columnwidth]{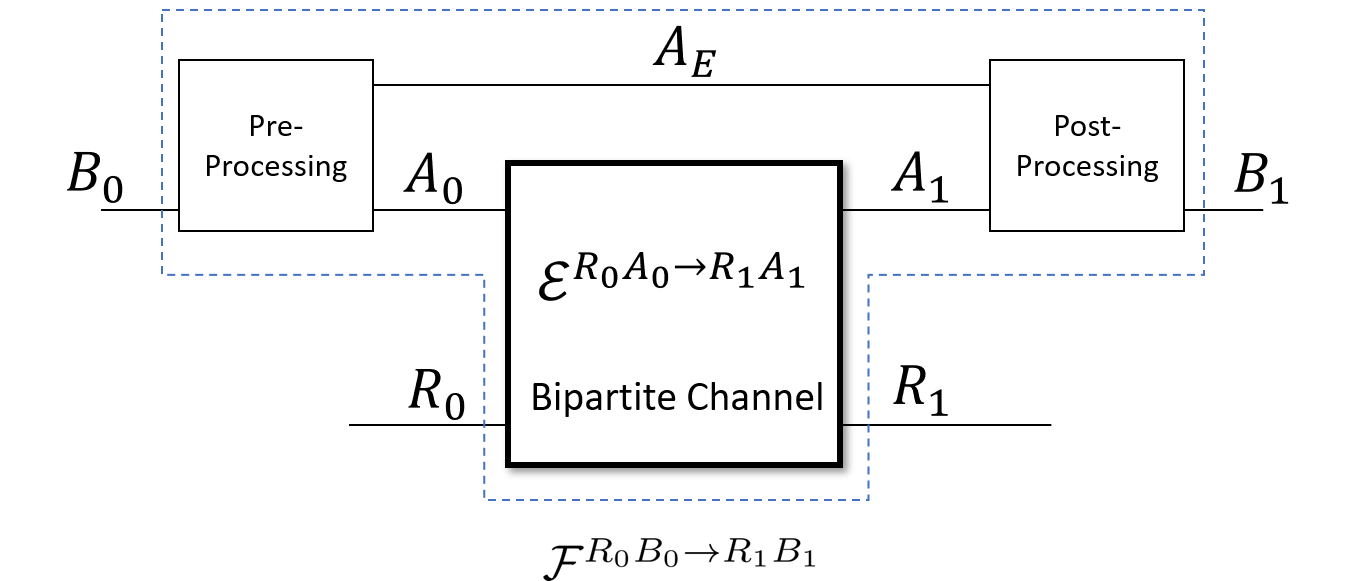}
    \caption{A local superchannel transforms $\mc{E}^{R_0A_0\to R_1A_1}$ into $\mc{F}^{R_0B_0\to R_1B_1}$.  An entanglement-breaking superchannel outputs a separable channel for every initial channel. }
    \label{fig:local-superchannel}
\end{figure}

Entanglement-breaking superchannels (EBSCs) generalize the class of entanglement-breaking channels (EBCs), a well-studied object within quantum information theory \cite{Horodecki2003EB}.  A channel $\mc{N}^{A_1\to B_1}$ is called entanglement breaking (EB) if $\id^{R_1}\otimes\mc{N}^{A_1\to B_1}(\rho^{R_1A_1})$ is separable for every $\rho^{R_1A_1}$ and all systems $R_1$.  That an EBC is a special case of an EBSC comes from the fact that every quantum state can be regarded as a quantum channel with a one-dimensional input.  An EBC $\mc{N}^{A_1\to B_1}$ can then be seen as an EBSC that locally transforms quantum channels with trivial input (see Fig.~\ref{fig:EBC}).  

\begin{figure}[t]
    \centering
    \includegraphics[width=0.45\columnwidth]{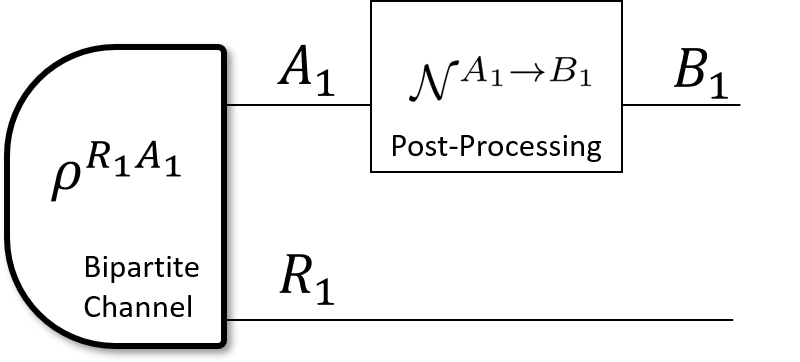}
    \caption{A state $\rho^{R_1A_1}$ represents a special type of bipartite channel, and a standard EBC can be seen as an EBSC post-processing map.}
    \label{fig:EBC}
\end{figure}

A central theorem is that every EBC can be realized by a measure-and-prepare protocol~\cite{Horodecki2003EB}.  That is, $\mc{N}^{A_1\to B_1}$ is EB if and only if  there exists a measurement described by a positive operator-valued measure (POVM) $\{F^{A_1}_k\}_k$ along with a family of states $\{\omega_k^{B_1}\}_k$ such that 
\begin{equation}
\label{Eq:m-p-c}
\mc{N}^{A_1\to B_1}(\rho^{A_1})=\sum_k \tr[F_k^{A_1}\rho^{A_1}]\omega_k^{B_1}    
\end{equation}
for all states $\rho^{A_1}$ of system $A_1$.  The interpretation is that the EBC $\mc{N}^{A_1\to B_1}$ can be implemented simply by first measuring system $A_1$ with POVM $\{F^{A_1}_{k}\}_k$ and then preparing the state $\omega_k^{B_1}$ contingent on outcome $k$.  A chief question of interest in this paper is whether there exists a similar type of implementation for EBSCs.  We find that obvious forms of EBSC implementation fail to capture the entire class of EBSCs, including the generalization of measure-and-prepare channels to superchannels.

The study of EBSCs falls within the broader research program of understanding dynamical quantum resources.  A quantum resource theory (QRT) describes a generic framework for isolating some particular feature of a quantum system, like entanglement or coherence, and analyzing how that feature, or ``resource'', behaves under a restricted set of operations \cite{Chitambar-2019a}.  Most attention has been previously devoted to studying static resources, that is, features that arise in particular states of a quantum system.  However, recently, significant progress has been made in developing the theory of \textit{dynamical} quantum resources, which refer to certain properties of quantum channels that are of interest for quantum information processing \cite{Theurer-2019a, Berk-2019a, liu2020operational, Liu-2019b, Gour-2019b}.  For example, in the QRT of entanglement for bipartite channels \cite{Gour2019ent, Bauml-2019a}, a channel's ability to distribute entanglement is a resource, and when separable processing is taken as free, an EBSC can be interpreted as a one-sided resource-erasing map.

The paper is organized as follows. In Section~\ref{sec:pre}, we fix notations and provide necessary preliminaries. In Section~\ref{sec:EBSC}, we formally define entanglement-breaking superchannels and give several equivalent characterizations. We also characterize a subset of EBSCs which allows for intuitive pre/post-processing realizations. In Section \ref{sec:superactivation}, we show how two copies of an EBSC can be combined to no longer be an EBSC, thereby demonstrating a type of superactivation.  In Section~\ref{sec:image}, we introduce the notion of CPTP and CP-complete image and consider what is the largest set of CPTP (CP) maps that can be generated through the action of EBSC, as well as two important subsets of EBSC.  In Section~\ref{sec:sidechannel-free}, we generalize EBSCs to superchannels that output $k$-non-entangling maps, and we connect these superchannels to the family of $k$-entanglement-breaking maps, the latter being a type of map that we introduce and thoroughly discuss. We summarize our results and conclude with some discussion in Sec.~\ref{sec:con}.


\section{Preliminaries} \label{sec:pre}
\subsection{Notations}
Throughout this paper we adopt most of the notations used in Ref.~\cite{Gour2019ent}. We use capital Latin letters $A, B, C$, \textit{etc.}~to denote physical systems, and $\mathcal{H}^A$, $\mathcal{H}^B$, $\mathcal{H}^C$, \textit{etc.}~to denote their corresponding Hilbert space. Sometimes we also use capital letters to denote Hilbert spaces for simplicity.   The collection of all bounded operators on system $A$ will be denoted by $\mathcal{B}(A)$, all Hermitian operators by $\textrm{Herm}(A)$, all positive operators by $\mathcal{P}(A)$, and all density matrices by $\mathcal{D}(A)$. The set of all linear maps from $\mathcal B(A)$ to $\mathcal B(B)$ will be denoted as $\mathcal{L}(A\to B)$, among which all completely-positive maps (CP) and all completely-positive and trace-preserving maps (CPTP) are denoted as $\textrm{CP}(A\to B)$ and $\textrm{CPTP}(A\to B)$, respectively.  A CPTP map is also called a quantum channel.

Since we are considering here dynamical resources, we will always assume that a system $A$ has an associated input and an output system, denoted by $A_0$, $A_1$, respectively. Therefore, we can use the shorthand notation $\mc{L}(A):=\mc{L}(A_0\to A_1)$, $\textrm{CP}(A):=\textrm{CP}(A_0\to A_1)$, \textit{etc.}~A linear map from $\mc L(A)$ to $\mc L(B)$ is called a supermap, and the set of all such supermaps will be denoted by $\mbb{L}(A\to B)$. The action of a supermap will be written as a square bracket, like $\Phi[\mc E]$, whereas the action of quantum channel will usually denoted by round brackets, like $\mc E(\rho)$.  

We use $\tilde A$ to denote a system with the same dimension of $A$, and $\phi_+^{\tilde{A}A}=\sum_{i,j=1}^{d_A}\op{ii}{jj}^{\tilde{A}A}$ is the unnormalized maximally entangled state on space $\tilde{A}A$.  For any linear map $\mc{E}\in\mc L(A_0\to A_1)$, its Choi matrix is defined as \cite{jamiolkowski1972linear, choi1975completely}
\begin{align}\label{eq:Choidef}
    J_{\mc{E}}^{A_0A_1}:=\id^{A_0}\otimes\mc{E}^{\tilde{A}_0\to A_1}(\phi_+^{A_0\tilde{A}_0}),
\end{align}
which establishes an isomorphism between $\mc L(A_0\otimes A_1)$ and $\mc L(A_0\to A_1)$, known as the Choi-Jamiolkowski isomorphism. The inverse of Eq.~\eqref{eq:Choidef} is
\begin{equation}
\mc E^{A_0\to A_1}(\rho^{A_0}) = \tr_{A_0}\left( \left((\rho^{A_0})^T\otimes I^{A_1}\right) J_{\mc E}^{A_0A_1}\right).
\end{equation}
Furthermore, $\mc E^{A_0\to A_1}$ is a CP map iff $J_{\mc E}^{A_0A_1}\ge 0$ (which means it is positive semidefinite), and $\mc E^{A_0\to A_1}$ is TP iff $\tr_{A_1}({J_{\mc E}^{A_0A_1}})=I^{A_0}$.

Throughout this paper, we denote the identity operator on state space $\mc{H}$ as $I$, the identity channel as $\id$ (\textit{i.e.}~$\id(\rho)=\rho$), and the identity superchannel as $\mbb{1}$ (\textit{i.e.}~$\mbb{1}[\mc E]=\mc E$). Hence, when we write $\mbb{1}\otimes\Theta[\mc E]$, it should be unambiguously understood as $(\mbb{1}\otimes\Theta)[\mc E]$ since $\mbb 1$ is a supermap, and should not be confused with $\id\otimes\Theta[\mc E]$.
 
\subsection{Separable and Entanglement-Breaking Maps}

We next review the meaning of separable and entanglement-breaking maps. 
\begin{definition} \cite{vedral1997quantifying, Barnum1998transmission}
A CP map $\mc E\in \mr{CP}(A_0R_0\to A_1R_1)$ is called $A_0A_1:R_0R_1$ \textbf{separable} if it can be written as $\mc E = \sum_k \Phi_k \otimes \Psi_k$ for some $\Phi_k\in \mr{CP}(A_0\to A_1)$ and $\Psi_k \in \mr{CP}(R_0\to R_1)$. 
\end{definition}
\noindent It is not difficult to see that $\mc{E}^{A_0R_0\to A_1R_1}$ is separable iff its Choi matrix
\begin{equation}
    J_{\mc{E}}^{A_0R_0A_1R_1}:=\id^{A_0R_0}\otimes \mc{E}^{\tilde{A}_0\tilde{R}_0\to A_1R_1}(\phi_+^{A_0\tilde{A}_0}\otimes \phi_+^{R_0\tilde{R}_0})
\end{equation}
is $A_0A_1:R_0R_1$ separable \cite{Cirac-2001a}.  This means that we can write $J_{\mc{E}}^{A_0A_1:R_0R_1}=\sum_kM_k\otimes N_k$ for some $M_k\in\mc{P}(A_0A_1)$ and $N_k\in\mc{P}(R_0R_1)$.  As alluded to in the introduction, separability of $\mc{E}^{A_0R_0\to A_1R_1}$ means that it cannot be used for distributing entanglement between Alice and Rachel, and this fact can be most easily established by examining its Choi matrix and using the identity $\mc{E}(\rho)=\tr_{A_0R_0}[(\rho^{A_0R_0})^T\otimes I^{A_1R_1}~J_\mc{E}^{A_0R_0A_1R_1}]$.

A close cousin to the separable maps are those that are entanglement breaking.
\begin{definition}\cite{Horodecki2003EB}\label{de:CEBC}
A CP map $\mc{N}\in\mathrm{CP}(A_1\to B_1)$ is called \textbf{entanglement breaking} (EB) if $\id^{R_1}\tens\mc{N}^{A_1\to B_1}(\rho^{R_1A_1})$ is separable for any $\rho^{R_1A_1}$ and any system $R_1$.  An EB map $\mc{N}$ is called an \textbf{entanglement breaking channel} (EBC) if it is also trace-preserving.
\end{definition}
\noindent The following provides different characterizations of EB maps and clarifies the relationship between EB and separable maps.
\begin{proposition}
\label{Prop:EBC}
    For a CP map $\mc{N}^{A_1\to B_1}$, the following are equivalent.
    \begin{enumerate}[label=(\Alph*)]
        \item $\mc{N}^{A_1\to B_1}$ is EB.
        \item Its Choi matrix $J_{\mc{N}}^{A_1B_1}$ is $A_1:B_1$ separable.
		\item $\mc{N}^{A_1\to B_1}(\rho)= \sum_k\tr(F^{A_1}_k\rho^{A_1})\omega_k^{B_1}$, for $F^{A_1}_k\ge 0$ and $\omega_k\in\mc D(B_1)$.     
%
        \item For any system $R$ and any bipartite channel $\mc{E}^{R_0B_1\to R_1A_1}$, the composition $\mc{F}^{R_0A_1\to R_1B_1}=\mc{N}^{A_1\to B_1}\circ\mc{E}^{R_0B_1\to R_1A_1}\circ\mc{N}^{A_1\to B_1}$ is a separable map.
    \end{enumerate}
\end{proposition}
\begin{proof}
Items (A)--(C) are standard results found in Ref.~\cite{Horodecki2003EB}.  From the form of Eq.~\eqref{Eq:m-p-c}, it is easy to see that $J_{\mc{F}}^{R_0A_1R_1B_1}$ is $R_0R_1:A_1B_1$ separable whenever $\mc{N}$ is EB, and so (A) $\Rightarrow$ (D).  Conversely, if (D) holds for all bipartite channels $\mc{E}^{R_0B_1\to R_1 A_1}$, then by considering the discard-and-prepare channel $\mc{E}(X^{R_0B_1})=\tr[X]\frac{1}{d_{A_1}}\phi_+^{R_1A_1}$ for $R_1\cong A_1$, we see that
separability of $\mc{F}$ implies that $\mc{N}$ is EB; hence (D) $\Rightarrow$ (A).
\end{proof}

\noindent\textbf{Remark:}  In Section \ref{Sect:CMPSC}, we will see that the channel $\mc{F}^{R_0A_1\to R_1B_1}$ constructed in (D) of this proposition is the output of a conditional prepare-and-measure superchannel $\mc{E}\mapsto \mc{N}\circ\mc{E}\circ\mc{N}$.


\subsection{Supermaps and Superchannels}

We next review the basic structure of superchannels.  The following definitions and theorems can be found in \cite{Gour2019channel}.
\begin{definition} \label{de:SC}\cite{Gour2019channel, Chiribella2008}
A supermap $\Theta\in \mathbb L(A\to B)$ is called a \textbf{superchannel} if both of the following are satisfied:
\begin{enumerate}[label=(\Alph*)]
    \item $\Theta$ is completely CP preserving: $\1^R \otimes \Theta [\mc E^{RA}]$ is completely positive for any CP map $\mc E$ and arbitrarily large system $R$.
    \item $\Theta$ is TP preserving: $\Theta[\mc E]$ is trace preserving for every TP map $\mc E$.
\end{enumerate}
\noindent We call $\Theta$ a \textbf{CP supermap} if condition (A) is satisfied.
\end{definition}


\noindent The notion of \textit{superchannel} defined here is equivalent to the \textit{deterministic supermap} defined in \cite{Chiribella2008}, and the term \textit{CP supermap} here is equivalent to \textit{probabilistic supermap} or simply \textit{supermap} in \cite{Chiribella2008}. In contrast to \cite{Chiribella2008}, here we use \textit{supermap} to refer to any linear map in $\mbb L(A\to B)$.


\medskip

The Choi matrix of a supermap $\Theta^{A\to B}$ is defined as
\begin{equation}
    \mc J_\Theta^{AB}=\sum_{a_0,a_1}  J^A_{\mc E_{a_0a_1}}\otimes J^B_{\Theta[\mc E_{a_0a_1}]},
\end{equation}
where $J^A_{\mc E_{a_0a_1}}$ and $J^B_{\Theta[\mc E_{a_0a_1}]}$ are the Choi matrices of $\mc E_{a_0a_1}\in\mc L(A)$ and $\Theta[\mc E_{a_0a_1}]\in\mc L(B)$ respectively, and $\left\{\mc E_{a_0a_1}\right\}_{a_0,a_1}$ is a complete orthogonal basis of $\mc L(A)$ whose action in the computational basis is given by
\begin{equation}
    \mc E_{ijkl}(\rho^{A_0}) = \bra i \rho^{A_0} \ket j \ket k \bra l^{A_1}. 
\end{equation}
Alternatively, $\mc J_\Theta^{AB}$ equals the Choi matrix of ${(\1^A\otimes\Theta^{\tilde A\to B})[\Phi_+^{A\tilde A}]}$, where
\begin{equation}
    \Phi_+^{A\tilde A} \equiv \sum_{a_0,a_1}\mc E^A_{a_0a_1}\otimes\mc E^{\tilde A}_{a_0a_1}.
\end{equation}
The action of $\Phi_+^{A\tilde{A}}$ can be expressed as
\begin{equation}
    \Phi_+^{A\tilde A}(\rho^{A_0\tilde{A_0}}) = \Tr(\rho\phi_+^{A_0\tilde A_0})\phi_+^{A_1\tilde A_1}, 
\end{equation}
and hence it can be viewed as a maximally entangled map, in analogy with $\phi_+$.  Note that $\Phi_+$ is not trace-preserving and hence not a quantum channel.  

From the Choi-Jamiolkowski duality, there is a one-to-one correspondence between supermaps and their Choi matrix. Hence, the linear spaces $\mbb L(A\to B)$, $\mc L(A_0\otimes A_1\to B_0\otimes B_1)$, $\mc L(A_0\otimes B_0\to A_1\otimes B_1)$, and $\mc L(A_1\otimes B_0\to A_0\otimes B_1)$ are all isomorphic. We have already seen that ${(\1^A\otimes\Theta^{\tilde A\to B})[\Phi_+^{A\tilde A}]}$ has the same Choi matrix as $\Theta$. Further define $\Delta_\Theta \in \mc L(A\to B)$ to be the unique map that satisfies $J_{\Delta_\Theta} = \mc J^{AB}_\Theta$, and $\Gamma_{\Theta}\in\mc L(A_1B_0\to A_0B_1)$ to be the unique map that satisfies $J_{\Gamma_\Theta} = \mc J^{AB}_\Theta$ . The properties of $\Theta^{A\to B}$ are then directly related to these three maps. Specifically, $\Theta$ is a CP supermap iff $(\1\otimes\Theta)[\Phi_+^{A\tilde A}]$, $\Gamma^{A_1B_0\to A_0B_1}_\Theta$ and $\Delta^{A\to B}_\Theta$ are CP. The condition when $\Theta$ becomes a superchannel is a little bit more involved and is given by the following lemma.


\begin{lemma} \cite{Chiribella2008, Gour2019channel}\label{le:superchannel}
Let $\Theta\in\mathbb L(A\to B)$. The following are equivalent.
\begin{enumerate}[label=(\Alph*)]
    \item $\Theta$ is a superchannel.
    \item The Choi matrix $\mc J^{AB}_{\Theta}\ge 0$ has marginals
    \begin{equation}\label{eq:marginal}
        \mc J^{A_1B_0}_{\Theta} = I^{A_1B_0};\quad \mc J_\Theta^{AB_0}=\mc J^{A_0B_0}\otimes u^{A_1}
    \end{equation}
    where $u^{A_1} \equiv \frac1{d_{A_1}}I^{A_1}$ is the normalized maximally mixed state on $A_1$.
    \item There exists a Hilbert space $\mc H_E$ with $d_E\le d_{A_0}d_{B_0}$, and two CPTP maps $\Gamma_\mathrm{pre}^{B_0\to A_0E}$ and $\Gamma_\mathrm{post}^{A_1E\to B_1}$ such that for all linear maps $\mc E^A$,
    \begin{equation}
        \Theta[\mc E^A]=\Gamma_\mathrm{post}^{A_1E\to B_1}\circ\left(\mc E^A\otimes\id^E\right)\circ\Gamma_\mathrm{pre}^{B_0\to A_0E}
    \end{equation}
\end{enumerate}
\end{lemma}

\noindent For a superchannel $\Theta^{A\to B}$, one can verify that the Choi matrix of $\Gamma_\mathrm{post}^{A_1E\to B_1}\circ\Gamma_\mathrm{pre}^{B_0\to A_0E}$ is $J^{AB}_{\Theta}$, and hence $\Gamma_\Theta^{A_1B_0\to A_0B_1} = \Gamma_\mathrm{post}^{A_1E\to B_1}\circ\Gamma_\mathrm{pre}^{B_0\to A_0E}$. On the other hand, by the Choi-Jamiolkowski duality we have
\begin{equation}\label{eq:choi_trans}
	J_{\Theta[\mc E]}^B = \Tr_A\left(\mc J_\Theta^{AB}\left((J_{\mc E}^A)^\mr T\otimes I^B\right)\right).
\end{equation}
Therefore, the CP map $\Delta_{\Theta}^{A\to B}$ has a property that it transforms the Choi matrix of a channel to another Choi matrix of a channel, \textit{i.e. $\Delta_\Theta(J^A_{\mc E}) = J^B_{\Theta[\mc E]}$}.

Pictorially, the calculation of the Choi matrix of a superchannel $\Theta^{A\to B}$ is shown in Fig~\ref{fig:Choi}.

\begin{figure}[!htb]
    \centering
    \includegraphics[width=0.6\columnwidth]{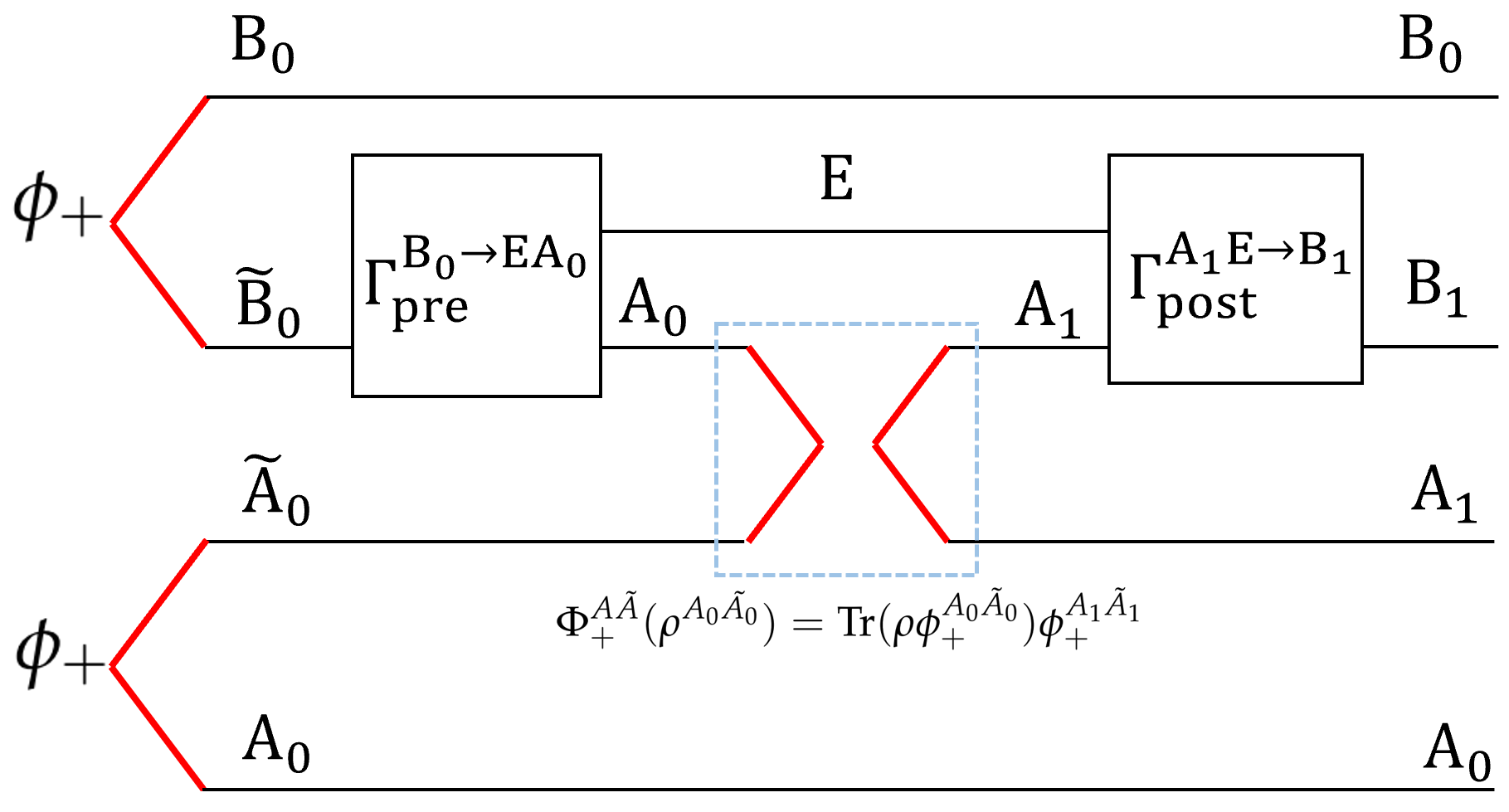}
     \caption{Circuit diagram for calculating the Choi matrix $\mc J_\Theta^{AB}$ of a superchannel $\Theta^{A\to B}$.} 
     \label{fig:Choi} 
\end{figure}



\section{Characterization and Realizations of EBSCs}
\label{sec:EBSC}

\subsection{Characterization of EBSCs}

With the background concepts in place from the previous section, we are now able to introduce the notion of entanglement-breaking supermap {and entanglement-breaking superchannel}.


\begin{definition}\label{de:EBSC}
A CP supermap $\Theta^{A\to B}$ is called an \textbf{entanglement-breaking supermap} if $\1^R \otimes\Theta[\mc E^{RA}]$ is a separable map for every $\mc{E}\in \textrm{CP}(RA)$, with $R$ being an arbitrary finite-dimensional system. If $\Theta^{A\to B}$ is furthermore a superchannel, then it is called an \textbf{entanglement-breaking superchannel} (EBSC).
\end{definition}

\noindent In this work we primarily focus on EBSCs which is a more physical object. However, in order to better understand the properties of EBSCs, in this and the next subsections, we will also look at EB supermaps. Our first result provides several equivalent characterizations of EB supermaps and EBSCs, in analogy to Proposition~\ref{Prop:EBC} for EBCs.


\begin{theorem}\label{th:EBSC}
For a CP supermap $\Theta^{A\to B}$, the following are equivalent.
\begin{enumerate}[label=(\Alph*)]
\item $\Theta^{A\to B}$ is an EB supermap.
\item $\1^{A} \otimes\Theta^{\tilde A\to B}[\Phi_+^{A\tilde A}]$ is separable for $\Phi_+^{A\tilde A}(\rho) = \Tr(\rho\phi_+^{ A_0\tilde{A}_0})\phi_+^{A_1\tilde{A}_1}$.
\item $\mc J_\Theta^{AB}$ is separable with respect to $A:B$.
\item $\Theta$ can be realized with pre/post-processing CP maps $\Gamma_\mathrm{pre}^{B_0\to A_0E}$ and $\Gamma_\mathrm{post}^{A_1E\to B_1}$ such that $\Gamma_\Theta^{A_1B_0\to A_0B_1} := \Gamma_\mathrm{post}^{A_1E\to B_1}\circ\Gamma_\mathrm{pre}^{B_0\to A_0E}$ is a separable map (with respect to $A_1A_0:B_0B_1$).
\item $\Delta_\Theta\in \mc L(A\to B)$ defined as the unique map with Choi matrix $J_\Theta^{AB}$ is an entanglement breaking map.
\end{enumerate}
If $\Theta$ is furthermore a superchannel, then the CP conditions in (D) can be strengthened to CPTP.
\end{theorem}


Sketch of the proof of Theorem~\ref{th:EBSC}: $(B),~(D),~(E)$ correspond to three different CP maps that have the same Choi matrix as $\Theta$, as discussed in Sec.~\ref{sec:pre}, so it is easy to show they are equivalent to $(C)$. $(A)\Rightarrow(B)$ is by definition of EBSCs. $(C)\Rightarrow (A)$ can be shown by considering the Choi matrix of $\1\otimes\Theta[\mc E]$ for any bipartite CP map $\mc E$.

\begin{proof}
$(A)\Rightarrow(B)$ is by Def.~\ref{de:EBSC} of EBSCs. $\1^{A} \otimes\Theta^{\tilde A\to B}[\Phi_+^{A\tilde A}]$ is separable if and only if its Choi matrix $J_{\1\otimes\Theta[\Phi_+^{A\tilde A}]}$ is separable with respect to $A:B$. Since $\mc J^{AB}_\Theta = J_{\1\otimes\Theta[\Phi_+^{A\tilde A}]}$ by definition, we see that $(B)\Leftrightarrow(C)$.

We now prove $(C)\Rightarrow(A)$. For any bipartite map $\mc E^{RA}\in\mathrm{CP}(R_0A_0\to R_1A_1)$, according to Eq.~\eqref{eq:choi_trans}, we have
\begin{align}
    J^{\tilde{R}B}_{\1\otimes\Theta[\mc E]}&=\Tr_{RA}\left(\mc J_{\1\otimes \Theta}^{RA\tilde{R}B}\left( (J^{RA}_{\mc E})^\mr T \otimes I^{\tilde{R}B} \right) \right)\\
    &= \Tr_{RA}\left(\left( \phi_+^{R_0\tilde{R}_0}\otimes\phi_+^{R_1\tilde{R}_1}\otimes \mc J_{\Theta}^{AB}\right)\left( (J^{RA}_{\mc E})^\mr T \otimes I^{\tilde{R}B} \right) \right)\\
    &= \Tr_A\left( \mc J^{AB}_\Theta \left((J^{\tilde{R}A}_{\mc E})^{\Gamma_A}\otimes I^{B}\right) \right),
\end{align}
where the superscript $\Gamma_A$ denotes partial transpose on system A. Since system $R$ is actually unchanged, let $\tilde{R} = R$ and rewrite this equation as
\begin{equation}\label{eq:prime}
    J^{RB}_{\1\otimes\Theta[\mc E]} = \Tr_A\left( \mc J^{AB}_\Theta \left(J^{RA}_{\mc E}\right)^{\Gamma_A} \right),
\end{equation}
where we omit the identity operator. Now $\mc J_\Theta^{AB}$ is separable, which means it can be written as
\begin{equation}\label{eq:sepDecomp}
    \mc J_\Theta^{AB} = \sum_{k=1}^{K} M^A_k\otimes N^B_k
\end{equation}
for some $M_k^A\in \mc P(A)$, $N_k^B \in \mc P(B)$, and positive integer $K$. Substituting into Eq.~\eqref{eq:prime}, we get
\begin{equation}
    J^{RB}_{\1\otimes\Theta[\mc E]} = \sum_k \Tr_A\left(M^A_k\left(J^{RA}_{\mc E}\right)^{\Gamma_A} \right)\otimes N^B,
\end{equation}
which means $J^{RB}_{\1\otimes\Theta[\mc E]}$ is separable, and hence $\1^R\otimes\Theta^{A\to B}[\mc E^{RA}]$ is separable. Since $\mc E^{RA}$ is an arbitrary bipartite CP map, we conclude that $\Theta$ is an EBSC, which completes the proof of $(C)\Rightarrow(A)$.

This has established the equivalence $(A)\Leftrightarrow(B)\Leftrightarrow(C)$. As for the equivalence of $(C),~(D),~(E)$, we know that the Choi matrix of $\Gamma_\Theta^{A_1B_0\to A_0B_1}$ and $\Delta_\Theta^{A\to B}$ are both $\mc J^{AB}_\Theta$, as shown in the last section. Therefore, the separability of $\mc J^{AB}_\Theta$ with respect to $A:B$ is equivalent to the separability of $\Gamma_\Theta^{A_1B_0\to A_0B_1}$ with respect to $A_1A_0:B_0B_1$, which is also equivalent to the fact that $\Delta^{A\to B}$ is entanglement breaking. 

The final remark on the case $\Theta$ being a superchannel is by Lemma~\ref{le:superchannel} that the pre/post-processing map of $\Theta$ can always be chosen to be CPTP (with $\Gamma_\Theta^{A_1B_0\to A_0B_1}$ unchanged).
\end{proof}

\begin{figure}[!htb]
    \centering
    \includegraphics[width=0.6\columnwidth]{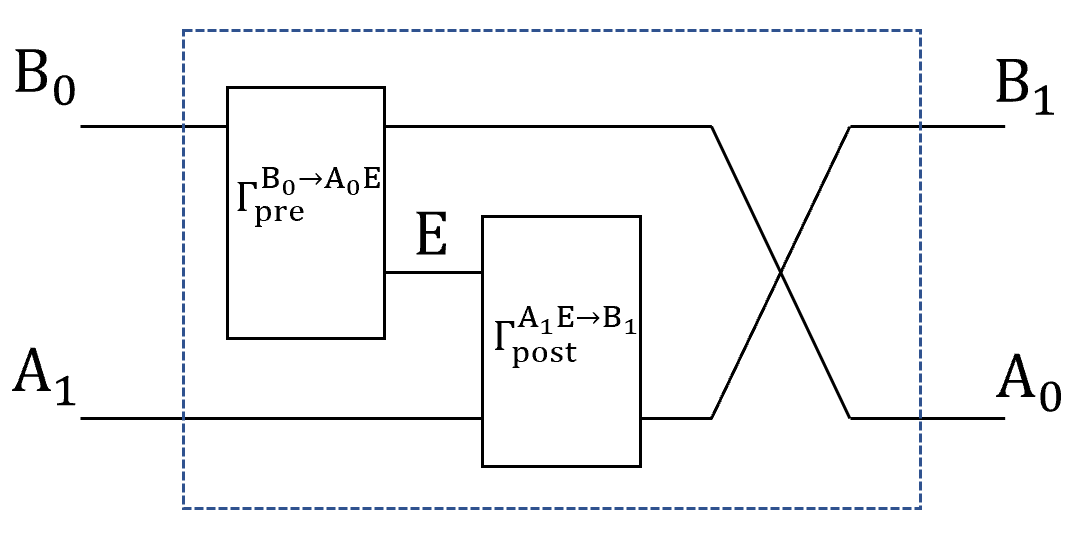}
     \caption{The map $\Gamma_{\Theta}^{A_1B_0\to A_0B_1}=\Gamma_\mathrm{post}^{A_1E\to B_1}\circ\Gamma_\mathrm{pre}^{B_0\to A_0E}$.  For any EB supermap $\Theta^{A\to B}$, Theorem~\ref{th:EBSC} (D) requires this map to be $A_0A_1:B_0B_1$ separable.} 
     \label{fig:prepost} 
\end{figure}

\subsection{EBSCs realized by EB pre/post-processing maps}

Theorem~\ref{th:EBSC} (D) provides a structural requirement for the pre/post-processing maps of an EBSC.  However, the separability on $\Gamma_\Theta$ only characterizes the \textit{concatenation} of the pre/post-processing maps, as depicted in Fig.~\ref{fig:prepost}.  
It remains unclear what constraints are placed on the pre/post-processing maps individually in order for the resulting superchannel to be entanglement-breaking.  Here, we obtain such a condition when the physical constraint of being trace-preserving is relaxed for the pre/post-processing maps.  Namely, we show that the pre/post-processing maps can be expressed as partly entanglement-breaking maps, as defined below.  While this provides some mathematical insight into the structure of EBSCs, it is not a statement about physical implementation due to the non-TP property of the maps.  When the pre/post-processing maps are required to valid quantum channels, it is unclear whether the pre/post-processing maps can always be taken as party-EB maps.  Instead, we show that entanglement needs to be carried by the side channel $E$ that connects the pre and post-processing channels for the implementation of certain EBSCs.

 

We begin by generalizing the definition of entanglement breaking for CP maps that have bipartite input or output.
\begin{definition}
\label{Defn:partly-EB-out}
A CP map $\mc E^{B_0\to A_0E}$ is called \textbf{partly entanglement breaking for output $A_0$} if $\id^R\otimes\mc E^{B_0\to A_0E}(\rho^{RB_0})$ is separable with respect to $A_0:RE$ for any $\rho\in\mc D(RB_0)$ and any system $R$ with finite dimension.
\end{definition}

\begin{definition}
\label{Defn:partly-EB-in}
A CP map $\mc E^{EA_1\to B_1}$ is called \textbf{partly entanglement breaking for input $A_1$} if $\id^{RS}\otimes\mc E^{EA_1\to B}(\sigma^{RE}\otimes\omega^{SA_1})$ is separable with respect to $S:RB_1$ for any $\sigma^{RE}\in\mc{D}(RE)$ and $\omega^{SA_1}\in\mc{D}(SA_1)$ for any system $R$, $S$ with finite dimension.
\end{definition}

The following two lemmas offer alternative characterizations of partly EB maps, similar to Proposition~\ref{Prop:EBC} for standard EB maps.
\begin{lemma}\label{le:EBout}
For a CP map $\mc E^{B_0\to A_0E}$, the following are equivalent:
\begin{enumerate}[label=(\Alph*)]
    \item $\mc E^{B_0\to A_0E}$ is partly EB for output system ${A}_0$.
    \item $\id^{B_0}\otimes\mc E^{\tilde{B}_0\to A_0E}(\phi_+^{B_0\tilde{B}_0})$ is separable with respect to $A_0:B_0E$.
    \item $\mc E^{B_0\to A_0E}$ can be written as
        \begin{equation}\label{eq:partlyEBout}
            \mc E^{B_0\to A_0E}(\rho^{B_0}) = \sum_k \mc E_k^{B_0\to E}(\rho^{B_0})\otimes \sigma^{A_0}_k,
        \end{equation}
for some $\mc E_k^{B_0\to E}\in\textrm{CP}(B_0\to E)$ and $\sigma^{A_0}_k\in\mc D(B)$ independent of the input.
\end{enumerate}
\noindent If $\mc E$ is furthermore CPTP, then $\{\mc E^{B_0\to E}_k\}_k$ in (C) is a quantum instrument, which means $\sum_k\mc E_k$ is CPTP.

\end{lemma}

\begin{proof}
$(A) \Rightarrow (B)$ is by definition and $(C)\Rightarrow(A)$ is trivial. It remains to prove $(B)\Rightarrow(C)$. Note $(B)$ says that the Choi matrix $J^{B_0A_0E}_{\mc E}$ is separable with respect to $A_0:B_0E$.  Hence we can write
\begin{equation}
    J^{B_0A_0E}_{\mc E} = \sum_k M^{B_0E}_{k}\otimes \sigma^{A_0}_{k}
\end{equation}
for positive operators $M^{B_0E}_k$, $\sigma^{A_0}_k$, and without loss of generality, $\tr(\sigma^{A_0}_k)=1$ for all $k$. Taking $\mc E_k^{B_0\to E}$ to be the unique CP map whose Choi matrix equals $M_K^{B_0E}$ completes the proof. When $\mc E$ is furthermore CPTP, we have $I^{B_0}=J^{B_0}_{\mc E} =\sum_k M^{B_0}_k$.  Hence, each $M^{B_0}_k$ is the Choi matrix for a CP map $\mc{E}_k^{B_0\to E}$ and their sum is trace-preserving.
\end{proof}
\noindent Similar conclusions hold for partly EB maps with a bipartite input.
\begin{lemma}\label{le:EBin}
For a quantum maps $\mc E^{EA_1\to B_1}$, the following are equivalent.
\begin{enumerate}[label=(\Alph*)]
    \item $\mc{E}^{EA_1\to B_1}$ is partly EB for input system $A_1$.
    \item $\id^{A_1E}\otimes\mc E^{\tilde{E}\tilde{A}_1\to B_1}(\phi_+^{A_1\tilde{A}_1}\otimes\phi_+^{E\tilde{E}})$ is separable with respect to $A_1:EB_1$.
    \item $\mc E^{EA_1\to B_1}$ can be written as
        \begin{equation}\label{eq:partlyEBin}
        \mc E^{EA_1\to B_1}(\rho^{EA_1}) = \sum_k \mc E_k^{E\to B_1}\left( \Tr_{A_1}(F^{A_1}_k\rho^{EA_1})\right).
    \end{equation} 
for some $\mc E_k^{E\to B_1}\in CP(E\to B_1)$ and $\{F^{A_1}_k\}\in\mc P(A_1)$.
\end{enumerate}
\noindent If $\mc E$ is furthermore CPTP, then $\{\mc E_k^{E\to B_1}\}_k$ in (C) forms a quantum instrument and $\{F^{A_1}_k\}_k$ forms a POVM.
\end{lemma}

\begin{proof}
The proof is very similar to the previous lemma. We only need to show $(B)\Rightarrow(C)$. For any $\rho\in\mc D(EA_1)$,
\begin{align}
    \mc E^{EA_1\to B_1}(\rho^{EA_1}) &= \Tr_{EA_1}(J_{\mc{E}}^{EA_1B_1}(\rho^{EA_1})^{\mr T}) \\
                            &= \sum_k \Tr_{EA_1}\left((F_k^{A_1})^{\mr T}\otimes N_k^{EB_1}(\rho^{EA_1})^{\mr T}\right).
\end{align}
Taking $\mc E_k^{E\to B_1}$ to be the unique CP map whose Choi matrix equals $N_k^{EB_1}$ completes the proof. When $\mc{E}$ is furthermore CPTP, we have that $I^{EA_1}=J_{\mc{E}}^{EA_1}=\sum_k (F_k^{A_1})^{\mr T}\otimes N_k^{E}$.  Hence $\{F_k\}_k$ defines a POVM, and $\{N_k^{EB_1}\}_k$ are Choi matrices for an instrument.
\end{proof}

Note that partly-EB maps are equivalent to bipartite separable maps with one trivial input/output system. That is, $\mc E^{B_0\to A_0E}$ is partly EB for output $A_0$ iff it is a separable map w.r.t. $A_0:B_0E$.  Similarly, $\mc E^{EA_1\to B_1}$ is partly EB for input $A_1$ iff $\mc E^{EA_1\to B_1}$ is separable w.r.t. $A_1:EB_1$.  This can be seen from item (B) in Lemmas \ref{le:EBout} and \ref{le:EBin}. 
One may also notice from Eq.~\eqref{eq:partlyEBout} and Eq.~\eqref{eq:partlyEBin} that partly EB channels take a ``measure-and-prepare'' form. But we must be careful here. Although Eq.~\eqref{eq:partlyEBout} can indeed be interpreted as preparing a state conditioned on the measurement outcome, Eq.~\eqref{eq:partlyEBin} is like preparing a CPTNI map conditioned on the measurement outcome, which in general cannot be done physically.  {Even when $\mc{E}^{EA_1\to B_1}$ is trace-preserving, the individual $\mc{E}_k^{E\to B_1}$ need not be proportional to a trace-preserving map.  This is the core reason why EBSCs differ from prepare-and-measure superchannels as will be discussed later.

To see the relevance of partly EB channels to our problem, let us split system $E$ in definitions \ref{Defn:partly-EB-in} and \ref{Defn:partly-EB-out} into two parts: $E=E_AE_B$ (see Fig.~\ref{fig:conjecture}).  In this case, $\mc E^{B_0\to A_0E_AE_B}$ is partly EB for output $A_0E_A$ iff it is a separable map w.r.t. $A_0E_A:B_0E_B$.  Similarly, $\mc E^{E_AE_BA_1\to B_1}$ is partly EB for input $A_1E_A$ iff $\mc E^{EA_1\to B_1}$ is separable w.r.t. $A_1E_A:B_1E_B$.  Then it is not difficult to show that every EB supermap can be realized with partly-EB pre/post-processing maps of this form.}
\begin{theorem}
\label{thm:new}
For a CP supermap $\Theta\in\mbb L (A\to B)$, the following are equivalent
\begin{enumerate}[label=(\Alph*)]
\item $\Theta^{A\to B}$ is an EB supermap.
\item There exists a CPTP map $\Gamma_\pr^{B_0\to A_0E_AE_B}$ that is partly EB for the output system $A_0E_A$, and a CP map $\Gamma_\post^{A_1E_AE_B\to B_1}$ that is partly EB for the input system $A_1E_A$, such that
\begin{equation}
\Theta^{A\to B}[\mc E^A]=\Gamma_\mathrm{post}^{A_1E_AE_B\to B_1}\circ\left(\mc E^A\otimes\id^{E_AE_B}\right)\circ\Gamma_\mathrm{pre}^{B_0\to A_0E_AE_B}
\end{equation}
\end{enumerate}
\end{theorem}

\begin{figure}[!htb]
    \centering
    \includegraphics[width=0.7\columnwidth]{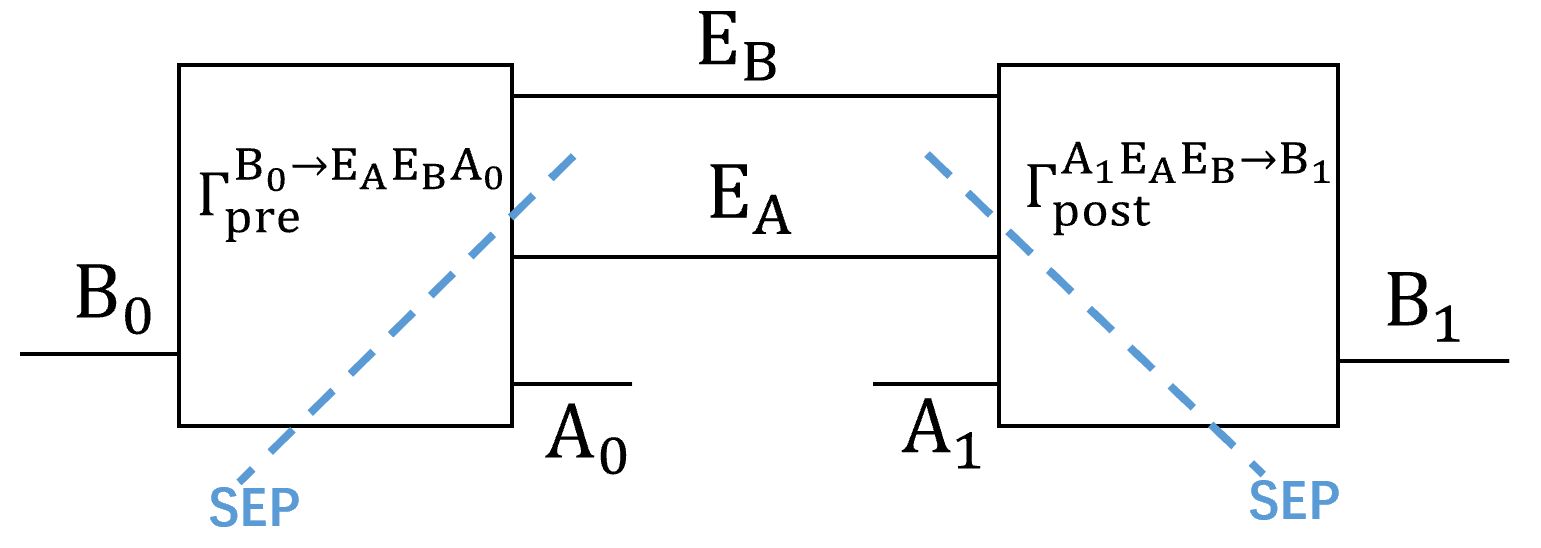}
    \caption{A general way to implement an EB supermap. $\Gamma_\mr{pre}^{B_0\to E_AE_BA_0}$ is partly EB for the output system $E_AA_0$ and $\Gamma_\mr{post}^{A_1E_AE_B\to B}$ is partly-EB for the input system $E_AA_1$.}
    \label{fig:conjecture}
\end{figure}

\begin{proof}
$(B)\Rightarrow (A)$: As is obvious from Fig.~\ref{fig:conjecture}, a CP supermap which can be realized as in (B) must be an EB supermap.

$(A)\Rightarrow (B)$: By Theorem~\ref{th:EBSC}, the Choi matrix of any EB supermap can be written as
\begin{equation}
\mc J_{\Theta}^{AB} = \sum_j M_j^{A}\otimes N_j^{B},
\end{equation}
for $M_j^{A}\in \mc P(A),~N_j^{B}\in \mc P(B)$. Let $E_A\cong A_0,~E_B\cong B_0$. Define the following pre/post-processing map
\begin{align}
\Gamma_\mr{pre}^{B_0\to A_0E_AE_B}(\rho) &= \id^{B_0\to E_B}(\rho)\otimes\frac{1}{d_{A_0}}\phi_+^{A_0E_A},\\
\Gamma_\mr{post}^{A_1E_AE_B\to B_1}(\rho) &= d_{A_0}\sum_j\mc N_j^{E_B\to B_1}\left(\tr_{E_AA_1}\left((M_j^{E_AA_1})^T\rho\right)\right),\label{eq:EB-post}
\end{align}
where $\mc N_j^{B_0\to B_1}$ is defined to be the unique CP map whose Choi matrix equals to $N_j^B$. Obviously these two maps are CPTP and CP respectively.

One can verify that the Choi matrix of $\Gamma_\mr{post}^{A_1E_AE_B\to B_1}\circ\Gamma_\mr{pre}^{B_0\to A_0E_AE_B}$ equals to $\mc J_\Theta^{AB}$, which means that this pair of pre/post-processing maps indeed implements $\Theta^{A\to B}$. On the other hand, by Lemma~\ref{le:EBout} and Lemma~\ref{le:EBin}, we see these two maps are partly-EB with respect to system $E_AA_0$ and $E_AA_1$, respectively. This completes the proof.
\end{proof}

The above theorem provides a structural condition on the pre/post-processing maps that can be used to construct any EB supermap. However, it fails to conclude that every EBSC can be realized by a pair of partly-EB pre/post-processing channels since the construction involves a non-TP map in Eq. \eqref{eq:EB-post}. Nevertheless, we conjecture that it is still possible to realize every EBSC by partly-EB pre/post-processing channels using a different construction than the one given in Theorem \ref{thm:new}.  We leave this conjecture as an open question to pursue elsewhere.  One difficulty in tackling this problem is that the splitting of side channels $E$ into two parts $E=E_AE_B$ makes the general structure of partly-EB pre/post-processing channels rather complex.  
Here we show that an implementation as in Fig.~\ref{fig:conjecture} but with $E_A$ removed is not strong enough to implement all EBSCs. Hence, implementing an EBSC will in general requires entanglement to be distributed across a memory side channel.  To see this, we first make a simple observation.


\begin{proposition}
{Suppose that $\Theta$ is a superchannel realized by a partly EB pre-processing map $\Gamma_{\pr}^{B_0\to A_0E}$ for $A_0$ and any \textrm{CPTP} post-processing map $\Gamma_{\post}^{EA_1\to B_1}$.  Then for any $\rho^{RB_0}$ and $\omega^{SA_1}$, the output state
\begin{equation}
    \id^{RS}\otimes\Gamma_{\post}^{EA_1\to B_1}\circ\Gamma_{\pr}^{B_0\to A_0E}\left(\rho^{RB_0}\otimes\omega^{SA_1}\right)
\end{equation}
is separable across $A_0:RB_1S$.  Similarly, if $\Theta$ is realized by a partly EB post-processing map for $A_1$, then this output state is separable across $A_0RB_1:S$.}
\end{proposition}
\begin{proof}
By definition, $\sigma^{A_0:RE}=\id^R\otimes\Gamma_{\pr}^{B_0\to A_0E}\left(\rho^{RB_0}\right)$ is separable across $A_0:RE$.  Then $\id^{S}\otimes\Gamma_{\post}^{EA_1\to B_1}(\sigma^{A_0:RE}\otimes\omega^{SA_1})$ will also be separable across $A_0:RB_1S$.  An analogous argument proves the second statement when the post-processing map is partly EB.
\end{proof}

From this proposition, it follows that if $\Theta^{A\to B}$ is an EBSC that can be realized by a partly EB pre-processing (resp. post-processing) map {for system $A_0$ (resp. system $A_1$)}, then we must have that $\mc J_\Theta^{AB}$ is both $A_0A_1:B_0B_1$ separable as well as $A_0:A_1B_0B_1$ (resp. $A_1:A_0B_0B_1$) separable.  It is not difficult to construct superchannels that fail to have this separability structure.
\begin{theorem}\label{Th:non-decomp}
{There exist EBSCs that cannot be realized using either a pre-processing channel that is partly EB for $A_0$ or a post-processing channel that is partly EB for $A_1$.  In other words, in Fig.~\ref{fig:conjecture}, input system $A_0$ must be entangled with some side channel $E_A$ and output system $A_1$ must couple with $E_A$ nonlocally.}

\end{theorem}
\begin{proof}
By the previous proposition, it suffices to construct pre/post-processing maps $\Gamma_{\pr}^{B_0\to A_0E}$ and $\Gamma_{\post}^{EA_1\to B_1}$ such that
\[\mc J_\Theta^{AB}=\id^{B_0A_1}\otimes\Gamma_{\post}^{E\tilde{A}_1\to B_1}\circ\Gamma_{\pr}^{\tilde{B}_0\to A_0E}\left(\phi^{+A_1\tilde{A}_1}\otimes\phi^{+B_0\tilde{B}_0}\right)\]
is $A_0A_1:B_0B_1$ separable, but neither $A_0:A_1B_0B_1$ nor $A_1:A_0B_0B_1$ is separable.  Let $d_{B_0}=d_{A_0}=d_{A_1}=2$ and $d_{B_1}=d_E=3$.  Define the isometry $\Gamma_{\pr}^{\tilde{B}_0\to A_0E}$ by
\begin{align}
    \ket{0}^{\tilde{B}_0}&\to\tfrac{1}{\sqrt{2}}(\ket{00}+\ket{11})^{A_0E},&\ket{1}^{\tilde{B}_0}&\to\ket{02}^{A_0E}.
\end{align}
Hence $\id^{B_0}\otimes \Gamma_{\pr}^{\tilde{B}_0\to A_0E}(\phi^{+B_0\tilde{B}_0})=\op{\tau}{\tau}$, where $\ket{\tau}=\frac{1}{\sqrt{2}}(\ket{00}+\ket{11})^{A_0E}\ket{0}^{B_0}+\ket{02}^{A_0E}\ket{1}^{B_0}$.  The post-processing map $\Gamma_{\post}^{E\tilde{A}_1\to B_1}$ is then a POVM with elements
\begin{align}
    \Pi_0&=\frac{1}{2}\phi^{+E\tilde{A}_1},\notag\\
    \Pi_1&=(I-\op{2}{2})^E\otimes I^{\tilde{A}_1}-\frac{1}{2}\phi^{+E\tilde{A}_1}\notag\\ \Pi_2&=\op{2}{2}^E\otimes I^{\tilde{A}_1}.
\end{align}
This leads to the Choi matrix
\begin{align}
    \mc{J}_\Theta^{AB}=\tfrac{1}{4}&\phi^{+A_0A_1}\otimes\op{00}{00}^{B_0B_1}+\tfrac{1}{2}\left(I^{A_0A_1}-\tfrac{1}{2}\phi^{+A_0A_1}\right)\otimes\op{01}{01}^{B_0B_1}\notag\\
    &+\op{0}{0}^{A_0}\otimes I^{A_1}\otimes\op{12}{12}^{B_0B_1}.
\end{align}
Clearly it is $A:B$ separable, and yet it is entangled for both parts $A_0$ and $A_1$ since projecting onto $\ket{00}^{B_0B_1}$ leads to the maximally entangled state $\tfrac{1}{4}\phi^{+A_0A_1}$.   
\end{proof}




\subsection{EBSCs realized by measure-and-prepare superchannels}

\label{Sect:CMPSC}

On the level of channels, entanglement-breaking is equivalent to measuring and preparing, as described in Eq.~\eqref{Eq:m-p-c}.  A natural question is whether the same holds for EBSCs. To this end, we first need a generalized notion of measurement that also applies to dynamical resources.  Recall that a quantum instrument is a collection of CP maps $\{\mc{E}_k\}_k$ whose sum $\sum_k\mc{E}_k$ is TP.  For a generalized measurement on state $\rho$ described by instrument $\{\mc{E}_k\}_k$, outcome $k$ occurs with probability $p_k=\tr[\mc{E}_k(\rho)]$ and the post-measurement state is $\mc{E}_k(\rho)/p_k$.  In Ref.~\cite{Burniston2019measure}, the authors introduce the concept of a \textit{quantum super-instrument}, which is a set of c-CPTNI (completely CP preserving and trace-non-increasing preserving) supermaps $\{\Theta_x\}$ such that $\sum_x \Theta_x$ is a superchannel.  A quantum super-instrument therefore describes in one sense how a quantum channel can be measured. 
It is proven in \cite{Burniston2019measure} that every super-instrument can be realized with a CPTP pre-processing channel and a post-processing quantum instrument, as shown on the left side of Fig.~\ref{fig:MeasureChannel}.
We are particularly interested in a special kind of super-instrument with trivial output channel, which is a generalization of the positive operator-valued measurement (POVM) to quantum channels, as shown on the right side of Fig.~\ref{fig:MeasureChannel}. Such object has been considered in the study of quantum channel discrimination  \cite{chiribella2008memory, Ziman2008PPOVM} where it is called a tester or a process-POVM, and in the study of quantum games \cite{gutoski2007toward} where it is called a co-strategy. 

\begin{figure}[!htb]
    \centering
    \includegraphics[width=0.9\columnwidth]{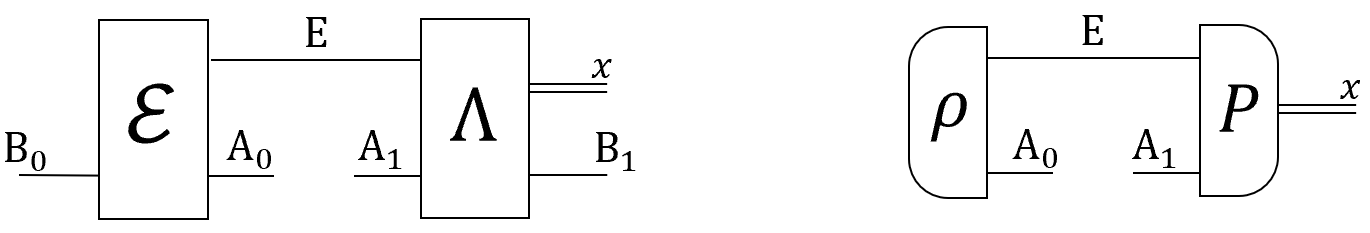}
    \caption{(Left) Implementation of a quantum super-instrument. A CPTP map $\mc E$ is performed as pre-processing while a quantum instrument $\Lambda = \sum_x \Lambda_x\otimes\ket x\bra x$ encompasses the post-processing. (Right) A POVM of quantum channels. The pre-processing is a quantum state $\rho$ and the post processing is a bipartite measurement channel $P = \sum_x\Tr_x(\cdot P_x)\otimes\ket x\bra x$ for POVM $\{P_x\}_x$.}
    \label{fig:MeasureChannel}
\end{figure}

With the concept of a channel POVM in place, we next combine it with a channel preparation step to obtain a measure-and-prepare superchannel.  This is depicted in Fig.~\ref{fig:MPSC} where the channel $\mc{F}^B_x$ is prepared contingent on the outcome $x$ of the proceeding channel POVM.
\begin{definition}\label{de:MPSC}
A superchannel $\Theta^{A\to B}$ is called a \textbf{measure-and-prepare superchannel} (MPSC) if it can be realized as
\begin{equation}
    \Theta^{A\to B}[\mc E^A] = \sum_x \Tr\left({\mr{id}^E\otimes\mc E^A(\rho^{A_0E}) P_x^{A_1E}}\right)\mc F_x^B,
\end{equation}
where $\rho^{A_0E}$ is some quantum state, $\{P_x^{A_1E}\}$ is some POVM, and $\{\mc F_x^B\}$ is a collection of CPTP maps.
\end{definition}

\begin{figure}[!htb]
    \centering
    \includegraphics[width=0.5\columnwidth]{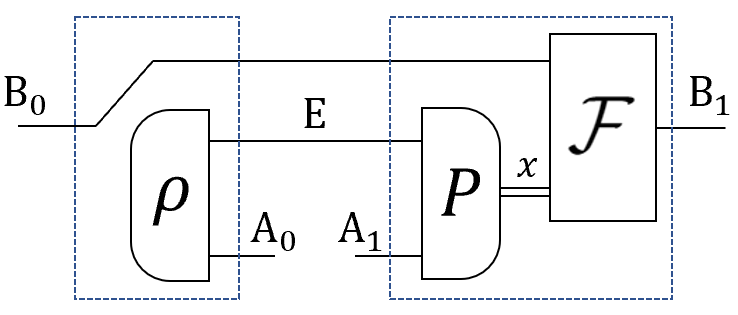}
    \caption{A measure-and-prepare superchannel (MPSC).}
    \label{fig:MPSC}
\end{figure}


There is a good physical motivation to consider MPSCs from the point of view of quantum memories.  An experimenter might want to transform an input channel $\mathcal{E}^A$ into a new channel that could be used at some later time.  For an MPSC, no quantum memory is needed to accomplish this goal.  Input channel $\mathcal{E}^{A}$ can be processed at time $t_0$ with output $x$ stored in classical memory.  Then at some later time $t_1$, system $B_0$ can be directly processed by map $\mathcal{F}_x$.  Note, for this interpretation to hold, it is crucial that $\mathcal{F}_x$ be trace-preserving.  

In the definition of an MPSC, the choice of pre-processing state $\rho^{E_1 A_0}$ provides an extra degree of freedom that is not present in the channel case.  Furthermore, the state $\rho^{E_1A_0}$ is independent of the input system $B_0$.  One may also want to apply a quantum instrument on $B_0$ and generate different pre-processing states conditioned on the measurement outcome of this instrument. We will refer to a superchannel having this type of structure as a controlled measure-and-prepare superchannel, or CMPSC for short.  Its rigorous definition is as follows.
\begin{definition}\label{de:CMPSC}
A superchannel $\Theta^{A\to B}$ is called a \textbf{controlled measure-and-prepare superchannel} (CMPSC) if it can be realized as
\begin{equation}
    \Theta^{A\to B}[\mc E^A] = \sum_{xy}  \Tr\left({\mr{id}^{E_1}\otimes\mc E^A(\rho_y^{A_0E_1}) {P_{x}^{A_1E_1}}}\right)\mc F_x^{E_2\to B_1}\circ\Lambda_y^{B_0\to E},
\end{equation}
where $\{\Lambda_y\}_y$ is a quantum instrument, $\rho_y^{A_0E}$ is some quantum state, {$\{P_{x}^{A_1E}\}_{x}$ is a POVM}, and $\{\mc F_x\}_x$ is a family of CPTP maps.
\end{definition}
The realization of a CMPSC is shown in Fig.~\ref{fig:CMPSC}.  

\begin{figure}[!htb]
    \centering
    \includegraphics[width=0.6\columnwidth]{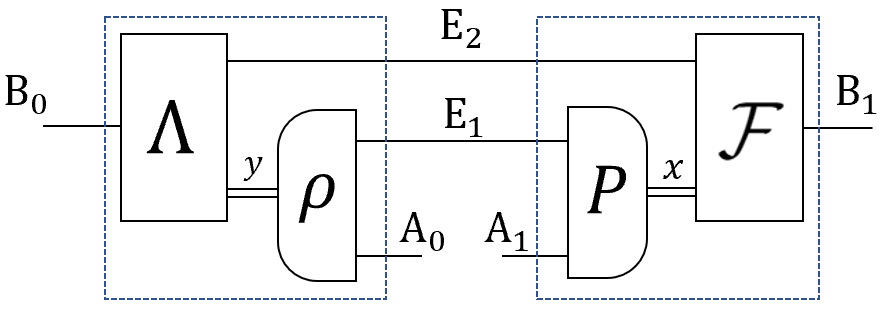}
    \caption{A controlled measure-and-prepare superchannel (CMPSC).}
    \label{fig:CMPSC}
\end{figure}

\noindent One can easily verify that every CMPSC is an EBSC. It is therefore natural to conjecture that every EBSC can be realized as an MPSC, similar to the case of EBCs. However, we will now show below that this conjecture fails to be true.


\begin{figure}[!t]
    \centering
    \includegraphics[width=0.6\columnwidth]{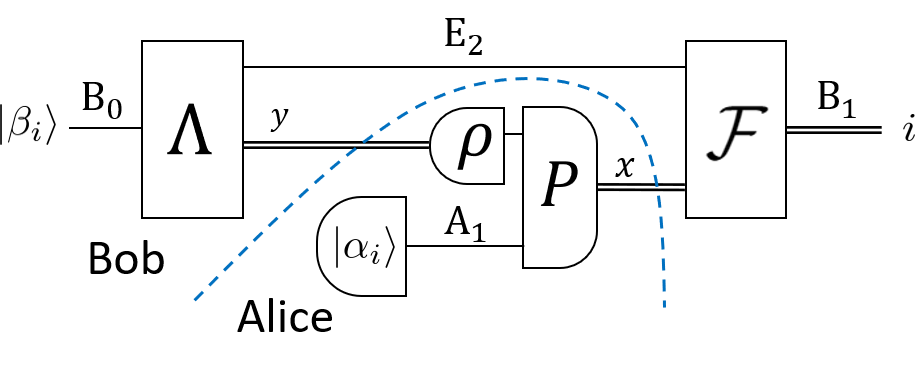}
     \caption{When $A_0$ is one-dimensional, a CMPSC reduces to a 2-way LOCC map on product states.  Since there are separable maps that cannot be implemented by 2-way LOCC, there exist EBSCs that are not CMPSCs.} 
    \label{fig:NLWE-2way}
\end{figure}

\begin{theorem}
\label{Thm:NLWE}
There exist EBSCs that are not CMPSC.
\end{theorem}
\begin{proof}
The construction uses the fact that there exist channels $\mc{E}^{E_2A_1\to B_1}$ that are partly EB for $A_1$, yet they are not implementable by local operations and classical communication (LOCC).  From the remark following Lemma~\ref{le:EBout}, a partly EB channel for input $A_1$ is an $A_1:E_2B_1$ separable channel.  Let  $\mc{E}^{E_2A_1\to B_1}$ be any such channel whose action on a collection of product states $\{\ket{\alpha_i}^{A_1}\otimes \ket{\beta_i}^{E_2}\}$ cannot be simulated by LOCC.  Such maps arise, for instance, when considering orthogonal product bases that can be distinguished by separable operations but not LOCC, a phenomenon famously known as \textit{nonlocality without entanglement} \cite{Bennett-1999a}.  Another example is the so-called double-trine ensemble \cite{Peres-1991a}, given by $\ket{\alpha_i}=\ket{\beta_i}=U^i\ket{0}$ for $i=0,1,2$, where $U=\text{exp}(-i\tfrac{\pi}{3}\sigma_y)$.  While an optimal minimum-error discrimination measurement on $\{\ket{\alpha_i}\otimes\ket{\beta_i}\}_{i=0,1,2}$ can be achieved by a separable map, it cannot be implemented by LOCC \cite{Chitambar-2013a}.  Let $\mc{E}^{E_2A_1\to B_1}$ be any such non-LOCC channel whose action is \[\mc{E}^{E_2A_1\to B_1}:\ket{\alpha_i}\ket{\beta_i}\mapsto\rho_{i}\]
for all $i$.  We can regard this as a post-processing map for a superchannel that converts each state $\ket{\alpha_i}$ into a QC channel having action $\ket{\beta_i}\to \rho_i$.  Since $\mc{E}^{E_2A_1\to B_1}$ is separable, it is clearly an EBSC.  However, with system $A_0$ being one-dimensional, the possible implementation by a CMPSC reduces to two-way LOCC, as shown in Fig.~\ref{fig:NLWE-2way}, which is not possible by construction.  Therefore the superchannel is not CMPSC.
\end{proof}



\section{Superactivation of EBSC}

\label{sec:superactivation}

Entanglement-breaking channels have the property that they are closed under tensor product.  That is, if $\mc{E}^{A_0\to A_1}$ and $\mc{N}^{R_0\to R_1}$ are both EB channels, then so is their tensor product, $[\mc{E}\otimes\mc{N}]^{A_0R_0\to A_1R_1}$.  In this section we consider whether an analogous result holds for EBSCs.  The problem has an added level of complexity when dealing with superchannels since the dynamic nature of channels allows them to be processed in different ways.  On the one hand, two copies of a superchannel can be used for parallel processing, which means that their pre/post-processing occurs simultaneously, as shown in Fig.~\ref{fig:EBSC-parallel}.  Alternatively, the two superchannels can implement a processing of channels in series, such that the output of one superchannel can be used as the input for the other, as shown in Fig. \ref{fig:EBSC-series}.  While the full theory of sequential processing can be described using the formalism of quantum combs \cite{Chiribella-2008a, Chiribella-2009a} or quantum strategies \cite{gutoski2007toward, Gutoski2012distance, Gutoski_PhD}, here we do not need to invoke these to demonstrate the generic superactivation of EBSC.

\begin{figure}[!t]
    \centering
    \includegraphics[width=0.6\columnwidth]{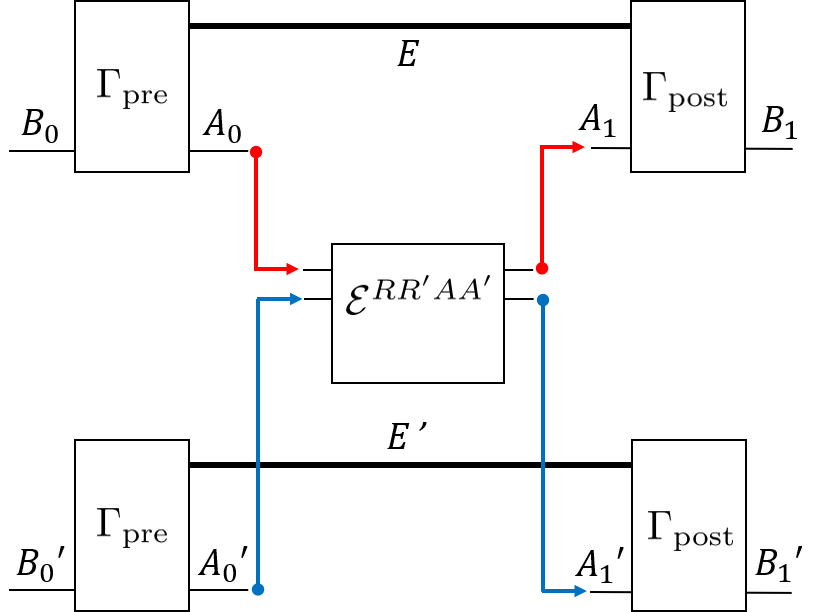}
     \caption{A parallel implementation of two EBSCs is still an EBSC.} 
     \label{fig:EBSC-parallel}
\end{figure}

\begin{theorem}
Suppose $\Theta^{A\to B}$ is an EBSC.  Then two copies of $\Theta^{A\to B}$ will be an EBSC when used in parallel.  However, if $\mc{J}_{\Theta}^{A_1B_0B_1}$ is entangled across $A_1B_0:B_1$ and $d_{A_1}\geq d_{B_0}$, then two copies of $\Theta^{A\to B}$ is no longer an EBSC when used in series.
\end{theorem}
\begin{proof}
When two copies of $\Theta^{A\to B}$ are used in parallel, they become a single superchannel $\Theta^{AA'\to BB'}$.  Its Choi matrix satisfies $\mc{J}_{\Theta}^{AA'BB'}=\mc{J}_{\Theta}^{AB}\otimes\mc{J}_{\Theta}^{A'B'}$, and since two copies of a separable operator is separable, it follows that $\mc{J}_{\Theta}^{AA'BB'}$ represents an EBSC.  On the other hand, suppose that each copy of $\Theta^{A\to B}$ can be used in series.  Let $d_{R_0}=d_{A_0}$, $d_{R_1}=d_{A_1}$, and define the channel $\mc{E}_0^{RA}$ that discards its input and prepares the state $\frac{1}{d_{B_0}}\phi_+^{R_1B'_0}$.  Likewise, let $\mc{E}_1^{R'A'}$ be the channel that discards and prepares $\frac{1}{d_{A_1}}\phi_+^{R'_1A_1'}$.   Then with the wiring shown in Fig.~\ref{fig:EBSC-series}, it follows that the resulting channel will be the state
\begin{equation}
   \Omega^{R_1R_1'B_1'}=d^{-1}_{B_0}d^{-1}_{A_1} \tr_{A_{0}'}\left(\id^{R'_1}\otimes\Gamma_{\post}^{EA_1'\to B_1'}\circ\id^{R_1}\otimes\Gamma_{\pr}^{B_0'\to A_0' E}\left(\phi_+^{R_1B_0'}\otimes\phi_+^{R_1'A_1'}\right)\right).
\end{equation}
However, this is proportional to $\mc{J}_{\Theta}^{A'_1B'_0B'_1}$ with $A_1'$ replaced by $R_1'$ and $B_0'$ replaced by $R_1$.  The $A_1'B'_0:B'_1$ entanglement implies that $\Omega^{R_1R_1'B_1'}$ is entangled between Rachel and Bob.  Hence  $\mc{E}_0^{RA}$ and $ \mc{E}_1^{R'A'}$ is transformed into an entangled channel by two copies of the EBSC $\Theta^{A\to B}$.
\end{proof}

\begin{figure}[!t]
    \centering
    \includegraphics[width=1\columnwidth]{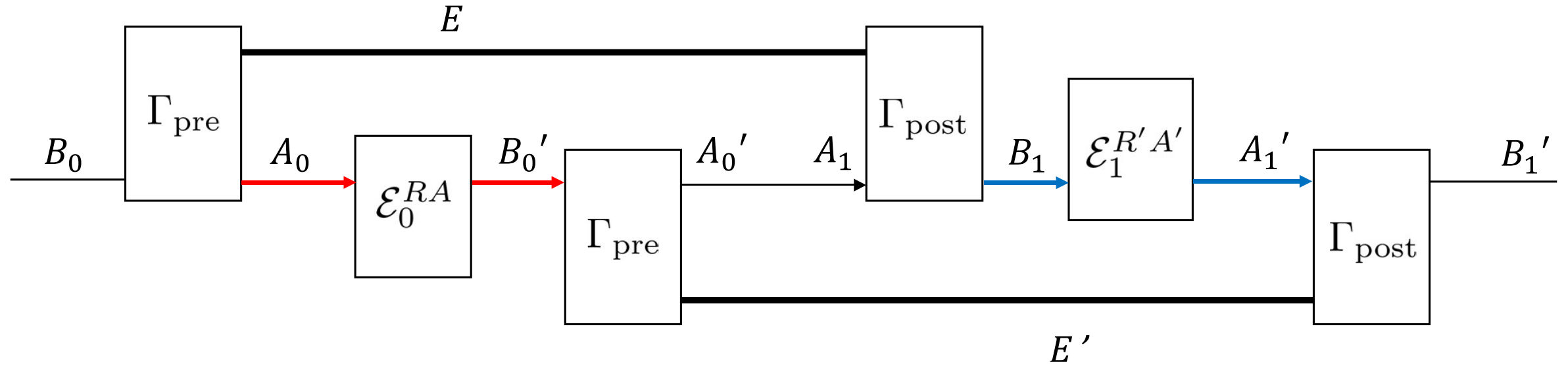}
     \caption{A series implementation of two EBSCs that can yield a non-separable channel.  The channels $\mc{E}_0^{RA}$ and $\mc{E}_1^{R'A'}$ discard their inputs and prepare maximally entangled states.} 
     \label{fig:EBSC-series}
\end{figure}

\section{Image of EBSC}\label{sec:image}

In Section \ref{sec:EBSC} we introduced EBSC and its two subsets MPSC and CMPSC. Here, we investigate the comparative powers of these operational classes.  In particular, we consider whether every bipartite separable channel can be generated by applying an EBSC on one part of some bipartite channel, and whether MPSC/CMPSC has the same channel-generation power as EBSC. For this purpose, we introduce the notion of \emph{CPTP-complete image}.

\begin{definition}
    Given a superchannel $\Theta^{A\to B}$, the \textbf{CPTP image} of $\Theta^{A\to B}$ is defined as
    \begin{equation}
        \mr{Im}_\Theta = \left\{ \mc G^B\in \mr{CPTP}(B) ~|~ \mc G^{B} = \Theta^{A\to B}[\mc E^{A}],~\mc E^A\in\mr{CPTP}(A)\right\}.
    \end{equation}
    The \textbf{CPTP-complete image} of $\Theta^{A\to B}$ is defined as
    \begin{equation}
        \mr{C}_\Theta = \bigcup_R \mr{Im}_{\1^R\otimes\Theta},
    \end{equation}
    which is the union of CPTP images $\1^R\otimes\Theta^{A\to B}$ for all systems $R$. The \textbf{CPTP-complete image} of a set of superchannels $\mc S$ is defined as
    $
        \mr{C}_\mc S = \bigcup_{\Theta\in \mc S}  \mr{C}_\Theta.
    $
\end{definition}
\noindent Similarly, we introduce the \emph{CP-complete image} as follows.
\begin{definition}
\label{de:CPim}
    Given a superchannel $\Theta^{A\to B}$, the \textbf{CP image} of $\Theta^{A\to B}$ is defined as
    \begin{equation}
        \mr{Im}^*_\Theta = \left\{ \mc G^B\in \mr{CP}(B) ~|~ \mc G^{B} = \Theta^{A\to B}[\mc E^{A}],~\mc E^A\in\mr{CP}(A)\right\}.
    \end{equation}
    The \textbf{CP-complete image} of $\Theta^{A\to B}$ is defined as
    \begin{equation}
        \mr{C}^*_\Theta = \bigcup_R \mr{Im}^*_{\1^R\otimes\Theta},
    \end{equation}
    which is the union of CP images $\1^R\otimes\Theta^{A\to B}$ for all systems $R$. The \textbf{CP-complete image} of a set of superchannels $\mc S$ is defined as
    $
        \mr{C}^*_\mc S = \bigcup_{\Theta\in \mc S}  \mr{C}^*_\Theta.
    $
\end{definition}

\noindent Alternatively, we can say that the CPTP image of $\Theta^{A\to B}$ is its image when the domain is restricted to CPTP maps, and likewise the CP image of $\Theta^{A\to B}$ is its image when the domain is restricted to CP maps.



With these definitions in place, a superchannel $\Theta^{A\to B}$ is an EBSC if and only if $\mr C^*_\Theta \subseteq \mr{SEP}^*$.  
Here we use SEP$^*$ to denote the set of separable CP maps while SEP denotes the set of all separable CPTP maps.  An interesting question is whether $\mr C_\mr{EBSC} = \mr{SEP}$ and $\mr C^*_\mr{EBSC} = \mr{SEP}^*$ hold, which will be the main concern of this section.

Let us first consider this question for EBCs.  Note that our definition of CPTP (CP) complete image also applies for channels since they are a special case of superchannels.  Specifically, for a channel $\mc{E}$, the set $\mr{Im}_{\mc{E}}$ (resp. $\mr{Im}^*_{\mc{E}}$) is the image of $\mc{E}$ when the domain is restricted to density operators (resp. positive operators).  In this case, it is easy to see that both $\mr C_\mr{EBC} = \mr{SEP}$ and $\mr C^*_\mr{EBC} = \mr{SEP}^*$ hold. Indeed for an arbitrary separable positive operator $\rho^{RB}=\sum_{i}p_i\op{\psi_i}{\psi_i}^R\otimes\op{\alpha_i}{\alpha_i}^B$, one need only consider the action of the EBC $\mc{E}^{A\to B}$ on the positive operator $\sigma^{RA}=\sum_{i}p_i\op{\psi_i}{\psi_i}^R\otimes\op{i}{i}^A$, where $\mc{E}^A(X)=\bra{i}X\ket{i}^A\op{\alpha_i}{\alpha_i}^B$, and $\sigma^{RA}$ is a density matrix if and only if $\rho^{RB}$ is.

In following part of this section, we set out to study the image of general EBSC. Regarding the CPTP-complete image, while we are unable to precisely characterize $\mr C_\mr{EBSC}$, we relate $\mr C_\mr{MPSC}$ and $\mr C_\mr{CMPSC}$ to LOCC channels of certain communication rounds.  As for the CP-complete image, we show that $\mr C^*_\mr{EBSC}=\mr{SEP}^*$ holds exactly. These results reveal some fundamental differences between channels and superchannels of physical significance, and we discuss this further at the end of this section.

\subsection{CPTP-complete Image of EBSC}




As noted above, one of our primary interests is determining whether $\mr C_\mr{EBSC}=\mr{SEP}$.  We consider here a special subset of the separable channels that can be generated by LOCC.  While LOCC is a notoriously difficult class of operations to analyze, here it will be sufficient to just consider one-round and two-round LOCC.  A precise definition of these operational classes is as follows.
\begin{definition}\label{de:LOCC2}\cite{Chitambar2014}
A channel $\mc E\in \mr{CPTP}(R_0B_0\to R_1B_1)$ is called a ($B\to R\to B$) two-round LOCC channel if it can be written as
\begin{equation}
    \mc E^{RB} = \sum_{ij}\Gamma^{R_0\to R_1}_{j|i}\otimes(\mc F^{B_2\to B_1}_{ij}\circ\Lambda^{B_0\to B_2}_i )
\end{equation}
for some quantum instrument $\{\Lambda_i^{B_0\to B_2}\}_i$ and family of channels $\{\mc F^{B_2\to B_1}_{ij}\}_{ij}$ for Bob, and a family of instruments $\{\Gamma^{R_0\to R_1}_{j|i}\}_{i,j}$ for Rachel.  An ($R\to B$) one-way LOCC channel is defined similarly except with the added condition that all the $\Lambda_i$ are trivial; \textit{i.e.}~$\mc{E}^{RB}=\sum_{j}\Gamma_j^R\otimes\mc{F}_j^B$.
\end{definition}




\begin{proposition}\label{prop:image}
The CPTP-complete images of EBSC, MPSC and CMPSC satisfy the following relation.
\begin{equation}\label{eq:image}
    \mr C_\mr{MPSC} = \mr{LOCC_1} \subsetneq \mr{LOCC_2} = \mr C_{\mr{CMPSC}} \subsetneq \mr C_\mr {EBSC} \subseteq \mr{SEP},
\end{equation}
where $\mr{LOCC_1}$ denotes the set of all ($R\to B$) one-round LOCC channels, $\mr{LOCC_2}$ denotes the set of all ($B\to R\to B$) two-round LOCC channels, and $\mr{SEP}$ denotes the set of all bipartite separable channels.
\end{proposition}


\begin{figure}[!b]
    \centering
    \includegraphics[width=0.7\columnwidth]{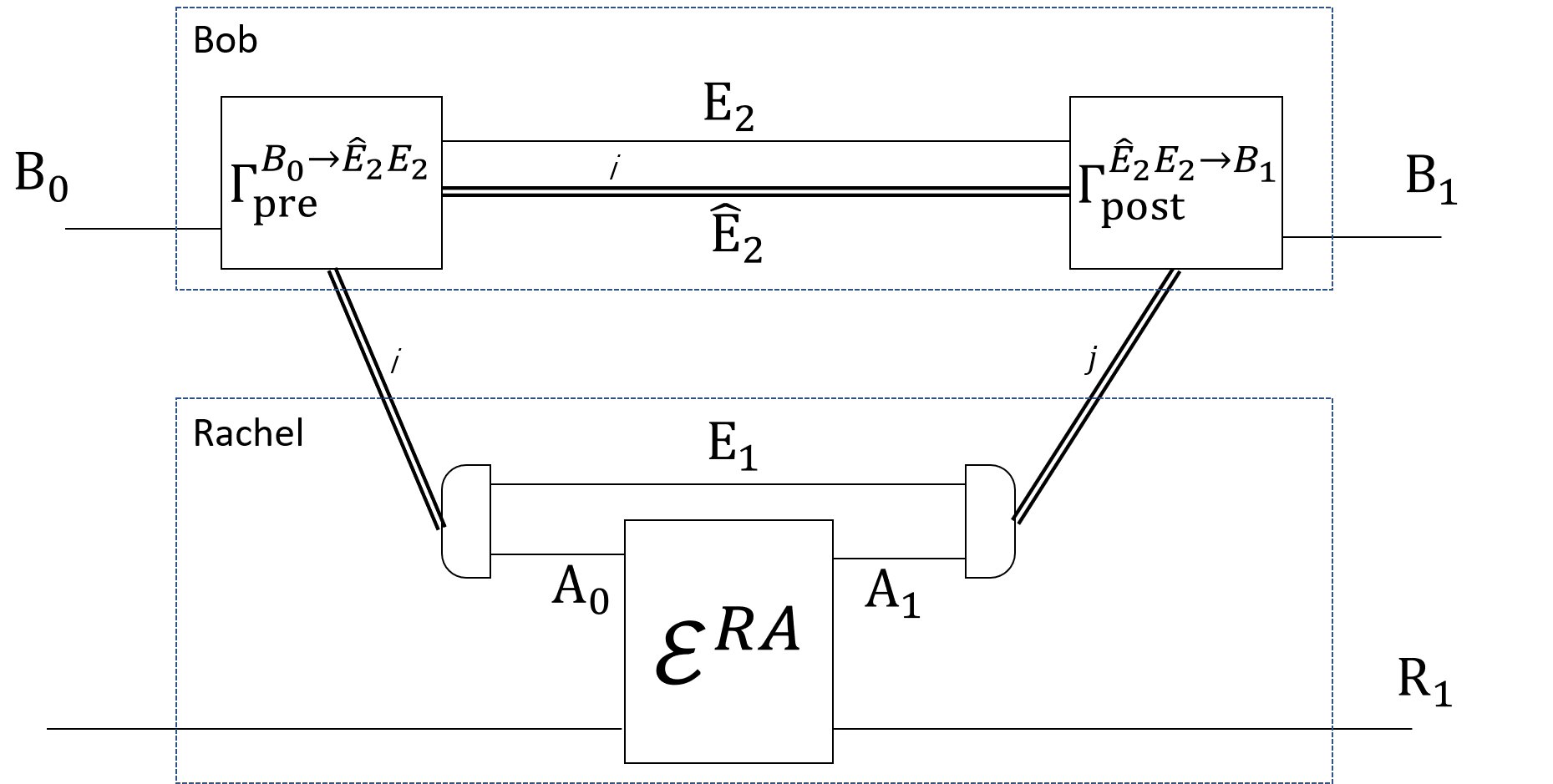}
     \caption{A CMPSC superchannel transforms any $\mc{E}^{RA}$ into a channel $\mc{F}^{RB}$ that can implemented by a two-round LOCC channel from Bob to Rachel.} 
     \label{fig:RA-to-RB-2}
\end{figure}

\begin{proof}
    For the second equality, first notice that if $\Theta$ is CMPSC, then $\1\otimes\Theta[\mc E]$ is a two-round LOCC channel for every quantum channel $\mc E^{RA}$, as can be seen from Fig.~\ref{fig:RA-to-RB-2}.  This establishes $\mr C_{\mr{CMPSC}}\subseteq \mr{LOCC}_2$. On the other hand, consider an arbitrary ($B\to R\to B$) two-round LOCC channel $\mc G$. By definition~\ref{de:LOCC2}, it can be written as
    \begin{equation}\label{eq:31b}
        \mc G^{RB} = \sum_{ij} \Gamma^R_{j|i}\otimes (\mc F^{B_2\to B_1}_{ij}\circ\Lambda^{B_0\to B_2}_i )
    \end{equation}
    for some CPTP maps $\mc F^B_{ij}$ and deterministic quantum instruments $\{\Lambda^B_i\}$, $\{\Gamma^R_{j|i}\}$. Set $B_2=E_2$ and construct a superchannel $\Theta$ with the following realization:
    \begin{align}
        \Gamma^{B_0\to A_0\hat{E}_2E_2}_\mr{pre}(\cdot) &= \sum_i \ket i\bra i^{A_0}\otimes\op{i}{i}^{\hat{E}_2}\otimes\Lambda^{B_0\to E_2}_i(\cdot),\\
        \Gamma^{A_1\hat{E}_2E_2\to B_1}_\mr{post}(\cdot) &= \sum_{i,j} \bra j \cdot \ket j^{A_1}\bra{i}\cdot\ket{i}^{\hat{E}_2}\mc F_{ij}^{E_2\to B_1}(\cdot).
    \end{align}
    Here, system $E_1$ is taken to be trivial. It is easy to see that $\Theta$ is a CMPSC by Def.~\ref{de:CMPSC}. Then consider the quantum channel
    \begin{equation}\label{eq:34b}
        \mc E^{RA}(\cdot)=\sum_{ij}\Gamma_{j|i}^R(\cdot)\otimes\bra i \cdot \ket i^{A_0}\otimes\ket j\bra j^{A_1}.
    \end{equation}
    One can immediately verify that
        $\mc G^{RB} = \1^R\otimes\Theta^{A\to B}[\mc E^{RA}],$
    which means $\mc G^{RB}\in\mr C_{\mr{CMPSC}}$. By the arbitrariness of $\mc G$ we have $\mr{LOCC_2}\subseteq\mr{C_{CMPSC}}$, hence $\mr{LOCC_2} = \mr C_{\mr{CMPSC}}$.
    
    For the first equality in Eq.~\eqref{eq:image}, we can similarly verify that $\mr{C_{MPSC}}\subseteq \mr{LOCC}_1$ from Fig.~\ref{fig:RA-to-RB-2} when the classical channel from Bob to Rachel is removed. For an arbitrary ($R\to B$) one-way LOCC channel $\mc G^{RB} = \sum_i \Gamma^R_i\otimes \mc F_i^B$ with CPTP maps $\mc F_i$ and quantum instrument $\{\Gamma_i\}_i$, construct a superchannel $\Theta^{A_1\to B}$ with trivial $A_0$ system as
    \begin{equation}
        \Theta^{A_1\to B}[\cdot] = \sum_i \bra i\cdot\ket i^{A_1}\otimes\mc F_i^B,
    \end{equation}
    and further define a quantum channel $\mc E^{R_0 \to R_1A_1} = \sum_i \Gamma^{R}_i\otimes \ket i\bra i^{A_1}$. Then we have $\1^R\otimes\Theta^{A_1\to B}[\mc E]=\mc G^{RB}$. It is obvious that $\Theta^{A_1\to B}$ is an MPSC, and so $\mc G \in \mr{C_{MPSC}}$. Combining the above results we get $\mr{C_{MPSC}} = \mr{LOCC}_1$.


As for $\mr{C_{CMPSC}} \subsetneq \mr{C_{EBSC}}$, recall the EBSC we constructed to prove Theorem~\ref{Thm:NLWE}. By applying that EBSC at an input noiseless channel $\id^{R_0\to A_1}$, we obtain a separable channel that cannot be implemented by LOCC. This complete the proof of proposition~\ref{prop:image}.
\end{proof}


The above proposition precisely characterizes the CPTP-complete image of MPSC and CMPSC. Also, it tells us $\mr{C_{EBSC}}$ contains all ($B\to R\to B$) two-round LOCC channels and some non-LOCC separable channels. It remains an open problem whether $\mr{C_{EBSC}}=\mr{SEP}$. We next turn to the CP-complete image.



\subsection{CP-complete Image of EBSC}


Interestingly, we can easily prove $\mr{C^*_{EBSC}} = \mr{SEP^*} $. We first present a proof by direct construction, and then later discuss how this result is related to the notion of stochastic LOCC (SLOCC) strategy. 

\begin{proposition}\label{prop:pimage}
The CP-complete image of EBSC and CMPSC satisfy
\begin{equation}
    \mr{C^*_{CMPSC}}=\mr{C^*_{EBSC}} = \mr{SEP^*},
\end{equation}
where $\mr{SEP^*}$ is the set of all bipartite separable CP maps.
\end{proposition}

\begin{proof}

By definition, $\mr{C^*_{EBSC}}\subseteq\mr{SEP^*}$. We only need to prove that $\mr{C^*_{EBSC}}\supseteq\mr{SEP^*}$. It is enough to show that $\mr{C^*_{EBSC}}$ contains all CPTNI separable maps, because every CP map can be normalized to be CPTNI by dividing a positive factor, and we can always multiply this factor to the input CP map.

 For any bipartite separable CPTNI map $\mc G^{RB}$, its Choi matrix can be written as $J^{RB}_{\mc G}=\sum_{k=1}^r P^{R_0R_1}_k\otimes Q^{B_0B_1}_k$ such that $P_k,~Q_k$ are positive operators. Without loss of generality we take $\Tr(P_k) = 1$.  Since $\mc G^{RB}$ is CPTNI, we have $\sum_k P^{R_0}_k\otimes Q^{B_0}_k\le I^{R_0B_0}$, and so $\sum_k Q^{B_0}_k\le I^{B_0}$. Let $F^{B_0B_1}$ be a positive operator satisfying $F^{B_0}=I^{B_0}-\sum_k Q^{B_0}_k$. Then for an $(r+1)$-dimensional system $A_0$ and system $A_1$ being one-dimensional, construct a supermap $\Theta$ whose Choi matrix is as follows.
    \begin{equation}\label{eq:34}
        \mc J^{AB}_\Theta = \sum_{k=1}^r \ket k\bra k^{A_0}\otimes Q_k^{B_0B_1} + \ket{r+1}\bra{r+1}^{A_0}\otimes F^{B_0B_1}.
    \end{equation}
Since $A_1$ is one-dimensional, it holds trivially that $\mc J^{AB_0}_\Theta = \mc J^{A_0B_0}_\Theta\otimes u^{A_1}$ and $\mc J^{A_1B_0}_\Theta = I^{A_1B_0}$, and hence $\Theta$ is a superchannel.   Furthermore, since $\mc J^{AB}_\Theta$ is separable with respect to $A:B$, it is also an EBSC, by Theorem~\ref{th:EBSC}.  Define a CP map $\mc E$ whose Choi matrix is given by
    \begin{equation}
    \label{Eq:TP-LOSR}
        J^{RA}_{\mc E} = \sum_{k=1}^{r} P_k^R\otimes\ket k\bra k^{A_0}.
    \end{equation}
Note the condition $\sum_{k}P_k^{R_0}=I$ need not be enforced here since we allow $\mc E$ to be non-trace-preserving.   Then,
    \begin{align}
        J^{RB}_{\1\otimes\Theta[\mc E]}&= \Tr_A\left( \mc J^{AB}_\Theta \left( J^{RA}_{\mc E} \right)^{\Gamma_A} \right) \\
                &= \sum_k P^R_k\otimes Q^B_k = J^{RB}_{\mc G},
    \end{align}
    hence $\mc G= \1\otimes\Theta[\mc E]$. By the arbitrariness of $\mc G$ and the argument before, we have $\mr{SEP^*\subseteq C^*_{EBSC}}$, hence $\mr{ C^*_{EBSC}=SEP^*}$.
    
    Notice that the EBSC $\Theta$ in the above proof is actually a CMPSC. To see this, take the pre/post-processing map to be 
    \begin{align}
        \Gamma_\mr{pre}^{B_0\to A_0E}&=\sum_{k=1}^r \Lambda_k^{B_0\to E}\otimes\ket k\bra k^{A_0} + \mc F^{B_0\to E}\otimes\ket{r+1}\bra{r+1}^{A_0},\\
        \Gamma_\mr{post}^{A_1E\to B_1}&= \id^{E\to B_1}\otimes\Tr_{A_1},
    \end{align}
    where $\Lambda_k^{B_0\to E}$ (resp. $\mc F^{B_0\to E}$) is the unique map with Choi matrix $Q_k^{B_0E}$ (resp. $F^{B_0E}$). It is easy to check $\Gamma_\mr{pre}^{B_0\to A_0E}$ is CPTP.
    Therefore, we have $\mr{C^*_{CMPSC} = C^*_{EBSC}= SEP^* }$. This complete the proof of Proposition~\ref{prop:pimage}.
\end{proof}

It is interesting to consider the physical interpretation of the superchannel and input CP used in the previous proof. In order to implement the CPTNI map $\mc G^{RB}=\sum_{k=1}^r P^R_k\otimes Q^B_k$, Bob first performs a quantum instrument $\{\Lambda_1,...,\Lambda_r,\mc F\}$ and sends the outcome (stored as classical information in system $A_0$) to Rachel. 
Rachel then acts as follows. If Bob's action is $\Lambda_k$, Rachel implements the CPTNI map with Choi matrix $P_k$ to complete the procedure. If Bob's action is $\mc F$, the protocol aborts.  Of course, it may only be possible for Rachel to implement the CPTNI map $P_k$ with some nonzero probability.  In this case, Rachel needs an extra round of classical communication to let Bob know whether her implementation is successful.  Bob's final CPTP map would then be the identity in the case that she succeeds, and some other fixed ``failure'' channel in the case that she does not.

The above procedure describes a general stochastic LOCC (SLOCC) protocol. It is shown in \cite{Dur2000} that every separable map can be implemented by an SLOCC protocol, which provides a rough explanation for why Proposition \ref{prop:pimage} holds. It also helps shed light on the physical significance of CP image, as defined in Definition \ref{de:CPim}.  This is related to stochastic quantum processes, which we discuss further in the next subsection.

\subsection{Discussion on CP Image}

Suppose that, for superchannel $\Theta^{A\to B}$ and CP maps $\mc G^B$ and $\mc E^A$ one has \begin{equation}\label{eq:cp}
    \mc G^B = \Theta^{A\to B}[\mc E^A].
\end{equation} 
If both $\mc G$ and $\mc E$ are CPTP, the physical interpretation of this equation is clearly deterministic channel conversion. What if they are not trace-preserving? Let us consider the case when $\mc G^B$ is CPTNI and $\mc E^{A}$ is a general CP map. We can always find a positive number $p\le 1$ such that $p\mc E$ is a CPTNI map, hence we can find another CPTNI map $\mc T$ to make up a quantum instrument $\{p\mc E,\mc T\}^{A}$. By applying $\Theta^{A\to B}$ to this instrument (see Fig.~\ref{fig:sc_instrument}) we get a new instrument $\{p\mc G,  \Theta[\mc T]\}^{B}$. This provides an implementation of the CPTNI map $\mc G$ with success probability $p$. Therefore, the physical interpretation of the above equation is probabilistic channel conversion. Hence the CPTP (resp. CP) image describes the channels (resp. CP maps) that can be deterministically (resp. probabilistically) generated by a superchannel. 

\begin{figure}[!htb]
    \centering
    \includegraphics[width=0.65\linewidth]{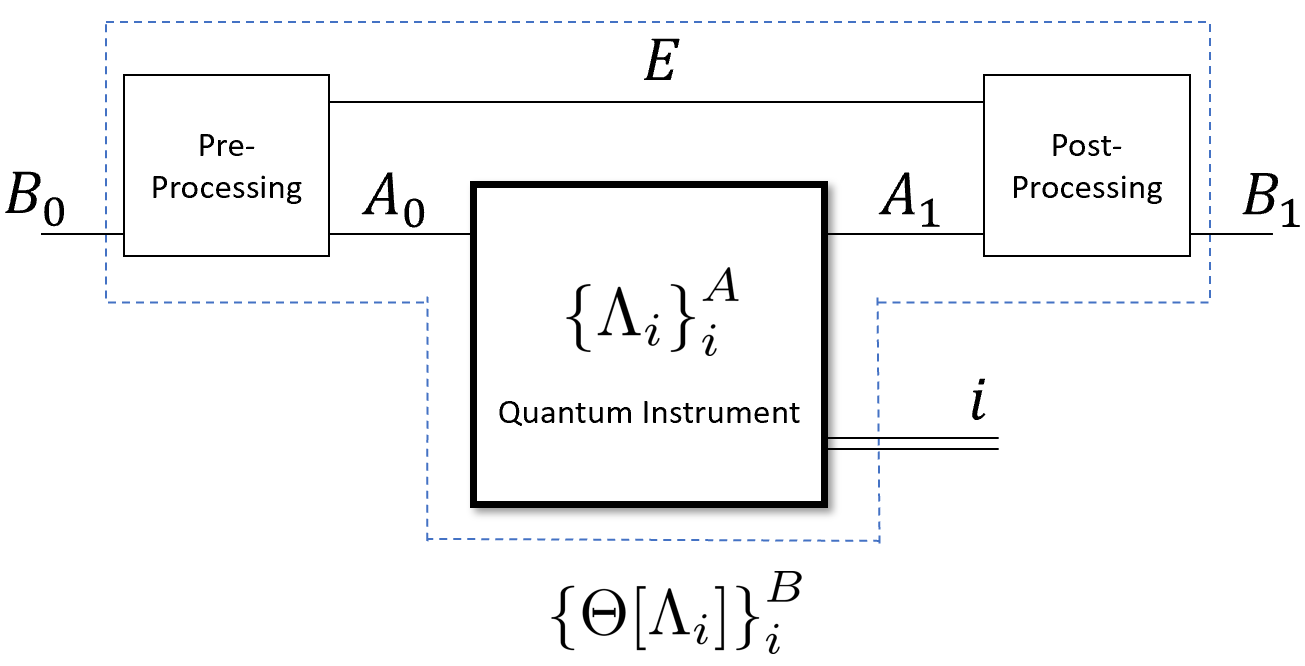}
    \caption{A quantum instrument $\left\{\Theta[\Lambda_i]\right\}_i^B$ resulting from applying superchannel $\Theta^{A\to B}$ at a bipartite quantum instrument $\left\{\Lambda_i\right\}_i^A$.
    }
    \label{fig:sc_instrument}
\end{figure}

Comparing Propositions \ref{prop:image} and \ref{prop:pimage}, which say that $\mr{C_{CMPSC}=LOCC_2}$ but $\mr{C^*_{CMPSC}=SEP^*}$, we see that there are CPTP maps that can be probabilistically implemented by CMPSC but not deterministically, since there are non-LOCC CPTP maps in $\textrm{SEP}\subset\textrm{SEP}^*$ that are not in $\mr{LOCC_2}$.  Hence, there are separable channels $\mc{G}^{RB}$ and valid superchannels $\Theta^{A\to B}$ for which $\mc G^{RB} = \1^R\otimes\Theta^{A\to B}[\mc E^{RA}]$ holds only for some CP non-TP map $\mc E^{RA}$. Such a phenomenon does not exist at the level of channels and states because we cannot transform an unnormalized state into a normalized one by a quantum channel.  This result shows that the structure of unnormalized maps is much more complex than the structure of unnormalized states, which is one of the reasons why QRTs of channels are often more challenging to formalize than QRTs of states (see also Ref.~\cite{Burniston2019measure} for a related discussion on this point).


\section{Alternative definitions of entanglement breaking}

\label{sec:sidechannel-free}

In the final section, we examine certain relaxations to the definition of entanglement breaking. To this end, we first consider relaxations to separable channels.  Recall that a bipartite channel $\mc{E}\in\mr{CP}(A_0R_0\to A_1R_1)$ is separable if $\sigma^{R'A'R_1A_1}=\id^{R'A'}\otimes\mc{E}^{A_0R_0\to A_1R_1}(\rho^{R'R_0}\otimes\omega^{A'A_0})$ is $R'R_1:A'A_1$ separable for all $\rho^{R'R_0}\otimes\omega^{A'A_0}$ and all systems $R'$ and $A'$.  In fact, it suffices to consider systems $R'$ and $A'$ of dimension $d_{R_0}$ and $d_{A_0}$, respectively.  A relaxation would be to require that $\sigma^{R'A'R_1A_1}$ is separable only for systems $R'$ and $A'$ of smaller dimension.

\begin{definition}
A bipartite channel $\mc{E}\in\mr{CP}(A_0R_0\to A_1R_1)$ is called \textbf{$(p,q)$-non-entangling} if $\id^{R'A'}\otimes\mc{E}^{A_0R_0\to A_1R_1}(\rho^{R'R_0}\otimes\omega^{A'A_0})$ is $R'R_1:A'A_1$ separable for all $\rho^{R'R_0}\otimes\omega^{A'A_0}$ with any systems $R'$ and $A'$ of dimension $d_{A'}=p$ and $d_{R'}=q$.
\end{definition}

\noindent A $(k,k)$-non-entangling channel is also called $k$-non-entangling in previous literature \cite{chitambar2020entanglement}. {We call a channel $(p,\text{complete})$-non-entangling if it is $(p,q)$-non-entangling for every positive integer $q$. A separable channel is then a $(\text{complete},\text{complete})$-non-entangling channel.}


\begin{definition}
A superchannel $\Theta^{A\to B}$ is called a \textbf{(p,q)-EBSC} if $\1^R \otimes\Theta^{A\to B}[\mc E^{RA}]$ is a $(p,q)$-non-entangling map ($p$ with $B$ and $q$ with $R$) for every $\mc{E}\in \textrm{CP}(RA)$, with $R$ being an arbitrary finite-dimensional system.
\end{definition}
\noindent Although the general structure of $(p,q)$-EBSCs could be rather complicated, 
we can obtain some decent results when restricting our attention to superchannels that have an implementation without a side-channel (see Def~\ref{de:sidefree} below), {and only $(k, \text{complete})$-EBSCs}. Before doing that, we first introduce a class of $k$-entanglement breaking channels in the following subsection.

\subsection{k-Entanglement Breaking Channel}


\begin{definition}\label{de:kEBC}\cite{christandl2019composed}
A quantum channel $\Lambda\in\mathrm{CPTP}(B_0\to B_1)$ is called a \textbf{k-entanglement breaking channel} ($k$-EBC) if $\id^R\tens\Lambda(\rho)$ is separable for any $\rho\in\mathcal{D}(R\otimes B_0)$, where $R$ is any $k$-dimensional Hilbert space.
\end{definition}

This class of $k$-EBC was first introduced in \cite{christandl2019composed} to study the condition when concatenated maps become entanglement breaking. Note that, this definition differs from the $k$-partially entanglement-breaking channel studied in Ref.~\cite{Chruscinski2006}, which is a channel $\Lambda^{B_0\to B_1}$ that satisfies $\mathrm{SN}(\id^R\tens\Lambda(\rho))\le k$ for any $\rho\in\mathcal{D}(R\tens B_0)$ with arbitrarily large system $R$, where $\mathrm{SN}$ denotes the Schmidt number.


For simplicity, in the following we only consider channels with the same input and output dimension, $d\geq 2$. A $k$-EBC is a completely-EBC whenever $k\ge d$~\cite{Horodecki2003EB}. For $1\le k<d$, the set of $k$-EBC is clearly a subset of ($k$+1)-EBC, and the set of all $1$-EBCs trivially equals $\mathrm{CPTP}(B_0\to B_1)$. To summarize, we have
\begin{equation}
    \text{completely-EBC}=\text{$d$-EBC}\subseteq\text{($d$-$1$)-EBC}\subseteq ... \subseteq \text{$1$-EBC} = \text{CPTP}.
\end{equation}
In Theorem~\ref{th:kEBC} below, we show that each of the above inclusions is strict. That is, there exists a $k$-EBC that is a not ($k$+1)-EBC for all $1\le k<d$.

%
\begin{theorem}\label{th:kEBC}
For all integers $k$ and $d$ with $1\le k< d$, there exists a $k$-EBC $\Lambda\in \mathrm{CPTP}(\mathbb{C}^d\to\mathbb{C}^d)$ that is not $(k+1)$-EBC.  
\end{theorem}

The proof of this theorem is inspired by Theorem~20 in \cite{chitambar2020entanglement}, which establishes the existence of $k$-non-entangling map that is not $(k+1)$-non-entangling. Interestingly, Theorem~20 in \cite{chitambar2020entanglement} turns out to be a corollary of Theorem~\ref{th:kEBC} presented here.

\begin{corollary}\cite{chitambar2020entanglement}
    For all integers $k$ and $d$ with $1\le k< d$, there exists a $k$-non-entangling map $\Lambda:\mathcal{L}(\mathbb{C}^d\tens\mathbb{C}^d)\to\mathcal{L}(\mathbb{C}^d\tens\mathbb{C}^d)$ that is not $(k+1)$-non-entangling.  
\end{corollary}
\begin{proof}
    The channel $\mathrm{SWAP}\circ(\Lambda\tens\Lambda)$ with $k$-EBC but not ($k$+1)-EBC $\Lambda$ is by definition $k$-non-entangling but not ($k$+1)-non-entangling map, for all $1\le k< d$.
\end{proof}

Now we set out to prove Theorem~\ref{th:kEBC}. In the following, we assume $\mathcal{H}_R = \mathbb{C}^k$ and $\mathcal{H}_{B_0} =\mathcal{H}_{B_1} = \mathbb{C}^d$ and $\Lambda\in\mathrm{CPTP}(B_0\to B_1)$. First notice that, in order to prove a channel $\Lambda$ is $k$-EBC, it suffices to show $\id_R\tens\Lambda(\varphi)$ is separable for all pure states $\ket\varphi\in \mathcal{H}_R\tens\mathcal{H}_{B_0}$, since we can then apply convexity to cover mixed states. Any such pure state can be written as
\begin{equation}
    \ket\varphi =  X\tens\id_d\ket{\phi_d^+}
\end{equation}
for some operator $X\in \mathcal{L}(\mathbb{C}^d\to\mathbb{C}^k)$. Let $X=\sum_{i=1}^r\lambda_i\op{\alpha_i}{\beta_i}$ denote a singular value decomposition of $X$, and define $X^{-1}:=\sum_{i=1}^r \lambda_i^{-1}\op{\beta_i}{\alpha_i}$ (with $\lambda_i>0$ for $i=1,\cdots r$).  Then $X^{-1}X=P=\sum_{i=1}^r\op{\beta_i}{\beta_i}$ is a projector of rank $r\leq k$.  Then since local operators preserve separability, we have that
 \[\id\otimes\Lambda[\op{\phi}{\phi}]=\id\otimes\Lambda[(X\otimes\mbb{I}_d)\phi^+_d(X\otimes {I}_d)^\dagger]\]
is separable if and only if
\begin{equation}
(X^{-1}\otimes\mbb{I})\id\otimes\Lambda[(X\otimes\mbb{I}_d)\phi^+_d(X\otimes I_d)^\dagger](X^{-1}\otimes I)^\dagger= (P\tens\id_d) J_\Lambda(P^\dagger\tens\id_d)
\end{equation}
is separable.  We summarize this observation in the following lemma.
\begin{lemma}\label{le:kEBC}
A channel $\Lambda\in \mathrm{CPTP}(B_0\to B_1)$ is $k$-EBC if and only if the operator
\begin{equation}
    (P\tens I_d) J_\Lambda(P^\dagger\tens I_d)
\end{equation}
is separable for all projectors $P$ with dimension no larger than $k$.
\end{lemma}

Now we introduce the Werner states \cite{Werner1989} which will play an important role in our construction of k-EBC. The Werner states are a family of states on $\mathbb{C}^d\tens\mathbb{C}^d$ which take the following form,
\begin{equation}
    \rho^\mathcal{W}_d(\beta)=\cfrac{1}{d^2-(\beta+1)}\left( I_{d^2}-\frac{\beta+1}{d}F_d \right),
\end{equation}
for $-(d+1)\le\beta\le d-1$, where $F_d = \sum_{ij}\ket{i}\bra{j}\tens\ket{j}\bra{i}$ is the swap operator (or the partial transpose of $\phi_d^+$).  A crucial observation is that the partial trace of a Werner state is $\frac{1}{d} I_{d}$, which means it is a valid Choi matrix of some CPTP quantum channel (with some normalization). Specifically, consider the following map,
\begin{equation}\label{eq:LambdaBeta}
    \Lambda_\beta(\rho_{B_0}) := d \Tr_{B_0}[\rho^\mathcal{W}_{d}(\beta)(\rho^{\mathrm{T}}_{B_0}\tens I_{B_1})].
\end{equation}
This map is clearly completely positive, and it is also trace-preserving since
\begin{equation}
    \begin{aligned}
        \Tr\left(\Lambda_\beta(\rho_{B_0})\right) &= \Tr[d\Tr_{B_1}\left(\rho^\mathcal{W}_{d}(\beta)\right)\rho_{B_0}]\\ &= \Tr( I_d \ \rho_{B_0})\\&=\Tr(\rho_{B_0}),
    \end{aligned}
\end{equation}
where the second equation is from our observation on the partial trace of Werner states. We conclude that $\Lambda_\beta\in\mathrm{CPTP}(B_0\to B_1)$.

Since $\rho_d^{\mc{W}}(\beta)$ is the Choi matrix of $\Lambda_\beta$ up to normalization, we can apply Lemma~\ref{le:kEBC} by studying the separability properties of $\rho_d^{\mc{W}}(\beta)$.  The following lemma is modified from Lemma~19 in \cite{chitambar2020entanglement}.
\begin{lemma}\label{le:werner}
Let $k$ and $d$ be integers such that $1\le k < d$ and let $-(d+1)\le\beta\le d-1$. The operator
\begin{align*}
    \rho^{\mathrm{proj}}_{RB_1}:=(P\tens I_d)\rho_d^\mathcal{W}(\beta)(P\tens I_d)    
\end{align*}
is separable for all projectors P with dimension no greater than $k$ if and only if $\beta\le(d-k)/k$.
\end{lemma}

\begin{proof}(Proof of lemma~\ref{le:werner})
We first require $P$ to be a $k$-dimensional projector in $\mathcal{L}(B_0\to B_1)$. By direct calculation, we have
\begin{align}
    \rho^\mathrm{proj}_{RB_1} &\propto  I_{kd} - \frac{\beta+1}{d}\left((P\otimes I_d) F_d (P^\dagger\otimes I_d)\right)\\
                            &= I_{kd} - \frac{\beta+1}{d}\left((P^*\otimes I_d) {\phi}_d^+ (P^{\mathrm{T}}\otimes I_d)\right)^\Gamma\\
                            &=  I_{kd} - \frac{\beta+1}{d}c\  \psi^\Gamma,
\end{align}
for some normalized pure state $\ket\psi=(P^*\otimes\id_d)\ket{{\phi}_d^+}/\sqrt{c}$, with normalizing factor
\begin{align}
    c &= \Tr\left( (P^*\otimes I_d){\phi}_d^+(P^\mathrm{T}\tens  I_d) \right) \\
      &= \Tr\left(\overline{P} P^\mathrm{T}\right) = \Tr\left( PP^\dagger\right)= k.
\end{align}
If $\beta>(d-k)/k$, then $\frac{\beta+1}{d}k>1$ and so
\begin{equation}
    \left(\rho^\mathrm{proj}_{RB_1}\right)^\Gamma = I_{kd} - \frac{\beta+1}{d}k\ \psi
\end{equation}
is obviously not positive. A non-positive partial transpose implies that $\rho^\mathrm{proj}_{RB_1}$ is entangled.  On the other hand, if $\beta\le(d-k)/k$, we apply Theorem~1 from \cite{Gurvits2002} which states that the operator $\id+A$ is separable for all Hermitian $A$ with $\|A\|_2\le1$, where $\|A\|_2 = \sqrt{\Tr A^\dagger A}$ is the Frobenius norm. Here, we have
\begin{equation}
    \left\| \frac{\beta+1}{d}k\ \psi^\Gamma \right\|_2 \le \|\psi^\Gamma\|_2 = \|\psi\|_2 = 1,
\end{equation}
since $\psi$ is a normalized pure state. We therefore conclude that $\rho^\mathrm{proj}_{RB_1}$ is separable.  This establishes that $\rho^\mathrm{proj}_{RB_1}$ is separable for all $k$-dimensional projectors $P$ if and only if $\beta\le(d-k)/k$.  Since any projector of dimension strictly less than $k$ can be implemented by first performing a $k$-dimensional projector and then projecting into a smaller subspace, it immediately follows that $\rho^\mathrm{proj}_{RB_1}$ is separable for all projectors $P$ of dimension no greater than $k$ if and only if $\beta\le(d-k)/k$. This completes the proof of Lemma~\ref{le:werner}.
\end{proof}

As a direct result of Lemma~\ref{le:kEBC} and Lemma~\ref{le:werner}, the channel $\Lambda_\beta$ constructed in Eq.~\eqref{eq:LambdaBeta} is $k$-EBC if and only if $\beta\le(d-k)/k$, for $1\le k<d$. This completes the proof of Theorem~\ref{th:kEBC}.

Finally, we note that, a special case of Theorem~\ref{th:kEBC} for $2$-EBCs is also established in \cite[Corollary~III.4]{christandl2019composed} using the Holevo-Werner maps similar as in the above proof.

 \subsection{Interplay between $k$-EBC and $(k,\text{complete})$-EBSC}

In this subsection, we discuss the interplay between generalized EBC and generalized EBSC. For simplicity, the systems $A$, $B$, we consider are both required to have $d$-dimensional input and output systems.  We will also restrict attention to the special class of superchannels that allow for a realization without the side-channel $E$.  


\begin{definition}\label{de:sidefree}
A superchannel $\Theta$ is said to be \textbf{without side-channel}, if it can be realized as
\begin{equation}\label{eq:sidefree}
    \Theta^{A\to B}[\mc E^A]= \Gamma_\mr{post}^{A_1\to B_1}\circ\mc E^A\circ\Gamma_\mr{pre}^{B_0\to A_0}
\end{equation}
for some CPTP maps $\Gamma_\mr{pre}$ and $\Gamma_\mr{post}$.
\end{definition}

\begin{corollary}
If an EBSC $\Theta^{A\to B}$ has a realization without side-channel as in Eq.~\eqref{eq:sidefree}, then both $\Gamma_\mr{pre}^{B_0\to A_0}$ and $\Gamma_\mr{post}^{A_1\to B_1}$ must be entanglement breaking channels.
\end{corollary}
\begin{proof}
This simply follows from Thm.~\ref{th:EBSC}.~(D) and the definition of EB channels.
\end{proof}

The following proposition discusses the relation between $k$-EBC and $(k,\text{complete})$-EBSC, for any positive integer $k\le d$. 

\begin{proposition}\label{prop:kEBSC}
For a superchannel $\Theta^{A\to B}$ without side-channel as in Eq.~\eqref{eq:sidefree}, the following are equivalent.
\begin{enumerate}[label=(\Alph*)]
    \item $\Theta^{A\to B}$ is a $(k,\text{complete})$-EBSC.
    \item $\Gamma_\mr{pre}^{B_0\to A_0}$ is a k-EBC, and $\Gamma_\mr{post}^{A_1\to B_1}$ is a completely-EBC.
\end{enumerate}
\end{proposition}
As a result, the realization of a $(k,\text{complete})$-EBSC without side-channel is shown in Fig.~\ref{fig:kd-EBSC}.
\begin{figure}[!htb]
    \centering
    \includegraphics[width=0.6\columnwidth]{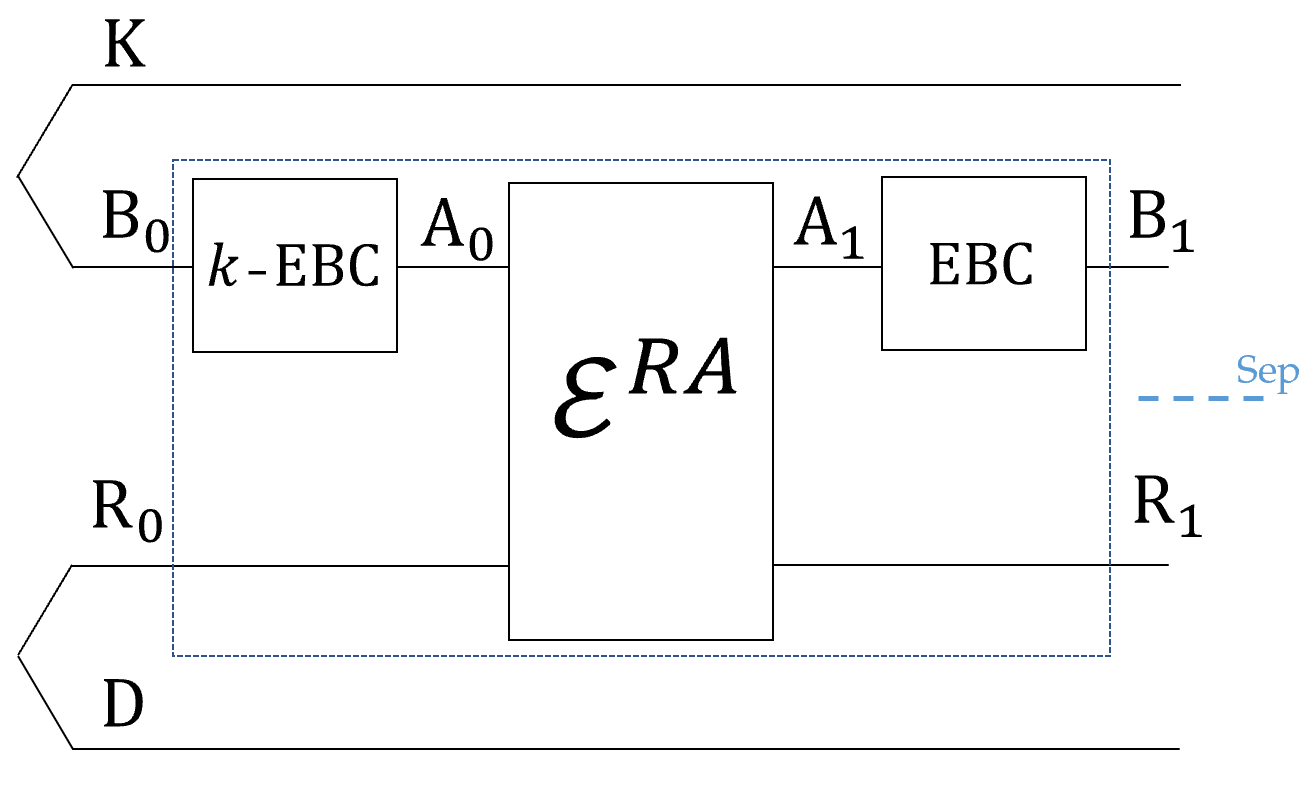}
    \caption{Realization of a $(k,\text{complete})$-EBSC without side-channel (in the blue dash box).} 
    \label{fig:kd-EBSC}
\end{figure}

\begin{proof}
$(A)\Rightarrow (B)$: Choose system $R$ to have $d_{R_0}=d_{R_1}=d$, Let $D$ be a $d$-dimensional system, $K$ be a $k$-dimensional system. Consider the CP map $\Phi_+^{RA}(\cdot)=\Tr((\cdot) \phi_+^{R_0A_0})\phi_+^{R_1A_1}$, the maximally entangled state $\phi_+^{R_0D}$, and an arbitrary quantum state $\rho\in\mc D(B_0K)$. We have
\begin{align}
    &\left(\mr{id}^{KD}\otimes(\1^R\otimes \Theta^{A\to B}[\Phi_+^{RA}])\right)(\rho^{B_0K}\otimes\phi_+^{R_0D})\\=~& 
    \Gamma_\mr{pre}^{B_0\to D}(\rho^{B_0K})\otimes\Gamma_\mr{post}^{A_1\to B_1}(\phi_+^{R_1A_1}).
\end{align}
Since $\Theta$ is $(k,\text{complete})$-EBSC, the above state must be separable with respect to $KB_1:DR_1$, which means both $\Gamma_\mr{pre}^{B_0\to D}(\rho^{B_0K})$ and $\Gamma_\mr{post}^{A_1\to B_1}(\phi_+^{R_1A_1})$ are separable. The former implies $\Gamma_\mr{pre}$ is $k$-EBC since $\rho$ is an arbitrary state on $B_0K$. The latter implies $\Gamma_\mr{post}$ is (completely) EBC by Proposition.~\ref{Prop:EBC} \cite{Horodecki2003EB}. This complete the proof of $(A)\Rightarrow (B)$.

$(B)\Rightarrow (A)$: Consider the circuit in Fig.~\ref{fig:kd-EBSC}. For any separable input $\rho^{B_0K}\otimes\varphi^{R_0D}$, it is easy to see the output state is $K:B_1:R_1D$ separable, which means $\Theta^{A\to B}$ is indeed (k,d)-EBSC.
\end{proof}

Since the hierarchy of $k$-EBC is non-trivial, there is also a non-trivial hierarchy of $(k,\text{complete})$-EBSC, for $1\le k \le d$. 

As a final comment, the series of study about when composed maps become entanglement breaking \cite{christandl2012ppt,lami2015entanglement, lami2016entanglement,rahaman2018eventually, 10.5555/3370201.3370205,christandl2019composed} may also be generalized to supermaps. As an example, the Schmidt number iteration theorem~\cite{christandl2019composed} can be directly applied here to get the following corollary.
\begin{corollary}
If $\Theta^{A\to A}$ is a (k,complete)-EBSC without side-channel and $M$ is an integer no less than $\ceil{\frac{d-1}{k-1}}$, then an $M$ times concatenation of $\Theta$ as $$\Theta^{M}\equiv\underbrace{\Theta\circ\Theta\circ\dots\circ\Theta}_{M~times} $$ is an EBSC.
\end{corollary}
\begin{proof}
The Schmidt number iteration theorem~\cite{christandl2019composed} states that a $\ceil{\frac{d-1}{k-1}}$ times concatenation of a $k$-EBC $\mc E^A$ becomes completely-EB. Combining it with Proposition~\ref{prop:kEBSC} gives this corollary.
\end{proof}
We leave the question when general composed superchannels become EBSC as an interesting future research direction.

\section{Conclusion and Discussion}\label{sec:con}

In this paper, we introduce and thoroughly study the notion of entanglement-breaking superchannels (EBSCs). These are objects that generalize and extend the standard notion of entanglement-breaking channels (EBCs) to ``higher-order'' quantum maps.  On the one hand, EBSCs (and EB supermaps) allow for relatively simple characterization via the Choi matrix, just like its channel counterpart. On the other hand, they can also exhibit some interesting properties, which make them much more complex than EBCs in many aspects.  
Firstly, we show that all entanglement-breaking supermaps can be decomposed into partly-EB pre/post-processing CP maps. While left open the question whether every EBSC can be implemented with partly-EB pre/post-processing channels, we show that a generic implementation must allow entangling the input system to the memory system of a superchannel.
Secondly, we show that EBSC is more general than measure-and-prepare superchannels (MPSCs), and even controlled-measure-and-prepare superchannels (CMPSCs), while for EBC these three classes coincide. Finally, we illustrate a super-activation phenomenon of EBSCs.

We further investigate which quantum channels can be generated using EBSCs, as well as the smaller classes of MPSCs and CMPSCs.  We show that the CPTP-complete image of MPSC/CMPSC equals one/two-round LOCC maps, respectively.  Although we are not able to precisely characterize the CPTP-complete image of EBSC, we show that its CP-complete image equals the collection of all separable maps. 
We argue that the notion of CP image captures some fundamental difference between channels and superchannels, and we hope these results might inspire new lines of investigation into probabilistic channel conversion.

In the final section of this paper, we establish a relationship between $k$-EBC, $k$-non-entangling channel and $(k,\text{complete})$-EBSC without side-channel. By generalizing the method of \cite{chitambar2020entanglement}, we show that all these three objects have a non-trivial hierarchy for $1\le k\le d$.   We remark that other alternative definitions of EBSC are also possible. One can require the output of EBSC to be not only separable, but also LOCC, or even LOSR (local operations and shared randomness). In this sense, our definition of MPSC (CMPSC) is just an example of one-round (two-round) LOCC-EBSC, but it remains unclear whether every one-round (two-round) LOCC-EBSC can be realized this way. 
In other words, whether $\mr{MPSC}\subseteq \mr{LOCC_1\text{-}EBSC}$ and $\mr{CMPSC}\subseteq \mr{LOCC_2\text{-}EBSC}$ are strict inclusions needs further investigation.

Our work provides a useful tool for the dynamical resource theory of quantum entanglement. Many results in entanglement theory based on EBCs can possibly be generalized to the dynamical resource theory with EBSC. For example, inspired by the resource theory of quantum memory where EBCs are a free resource \cite{Rosset2018memory}, one can consider the ability to faithfully store a quantum operation, perhaps call it a ``super-memory'', where EBSC may serve as free resource.  We leave this for future work.  Also, since we have characterized the Choi matrix of an EBSC, it is straightforward to calculate a robustness-type quantity with respect to it, similar to what has been done in \cite{Xiao2019memory}. We anticipate there being other applications of EBSC within the study of dynamical quantum resource theories.

There are some problems left open in our work. The first is whether the deterministic image of EBSC equals the set of all separable channels, namely whether $\mr{C_{EBSC}=SEP}$ or not. Currently we only know that all two-round LOCC and some non-LOCC separable channels are in $\mr{C_{EBSC}}$. The second is whether every EBSC can be realized as in Fig.~\ref{fig:conjecture} with the pre/post-processing maps being CPTP. Answering these questions will help us better understand the intricate structure of EBSCs.


\begin{acknowledgments}
We thank the anonymous reviewers for their helpful suggestions to improve this paper. SC acknowledges the support from the Tsinghua Scholarship for Undergraduate Oversea Studies.  EC is supported by the National Science Foundation (NSF) Award No. 1914440.
\end{acknowledgments}

\bibliographystyle{myunsrtnat}

\bibliography{EBSuper}

\begin{thebibliography}{41}
\providecommand{\natexlab}[1]{#1}
\providecommand{\eprint}[1]{\texttt{#1}}
\expandafter\ifx\csname eprintstyle\endcsname\relax
  \providecommand{\doi}[1]{doi: #1}\else
  \providecommand{\doi}{doi: \begingroup \eprintstyle{rm}\Url}\fi

\bibitem[Cirac et~al.(2001)Cirac, D\"ur, Kraus, and Lewenstein]{Cirac-2001a}
J.~I. Cirac, W.~D\"ur, B.~Kraus, and M.~Lewenstein.
\newblock Entangling operations and their implementation using a small amount
  of entanglement.
\newblock \emph{Phys. Rev. Lett.}, 86:\penalty0 544--547, Jan 2001.
\newblock \doi{10.1103/PhysRevLett.86.544}.

\bibitem[Chiribella et~al.(2008{\natexlab{a}})Chiribella, D'Ariano, and
  Perinotti]{Chiribella2008}
G.~Chiribella, G.~M. D'Ariano, and P.~Perinotti.
\newblock Transforming quantum operations: Quantum supermaps.
\newblock \emph{{EPL} (Europhysics Letters)}, 83\penalty0 (3):\penalty0 30004,
  jul 2008{\natexlab{a}}.
\newblock \doi{10.1209/0295-5075/83/30004}.

\bibitem[Horodecki et~al.(2003)Horodecki, Shor, and Ruskai]{Horodecki2003EB}
Michael Horodecki, Peter~W. Shor, and Mary~Beth Ruskai.
\newblock Entanglement breaking channels.
\newblock \emph{Reviews in Mathematical Physics}, 15\penalty0 (06):\penalty0
  629--641, 2003.
\newblock \doi{10.1142/S0129055X03001709}.

\bibitem[Chitambar and Gour(2019)]{Chitambar-2019a}
Eric Chitambar and Gilad Gour.
\newblock Quantum resource theories.
\newblock \emph{Rev. Mod. Phys.}, 91:\penalty0 025001, Apr 2019.
\newblock \doi{10.1103/RevModPhys.91.025001}.

\bibitem[Theurer et~al.(2019)Theurer, Egloff, Zhang, and Plenio]{Theurer-2019a}
Thomas Theurer, Dario Egloff, Lijian Zhang, and Martin~B. Plenio.
\newblock Quantifying operations with an application to coherence.
\newblock \emph{Phys. Rev. Lett.}, 122:\penalty0 190405, May 2019.
\newblock \doi{10.1103/PhysRevLett.122.190405}.

\bibitem[Berk et~al.(2019)Berk, Garner, Yadin, Modi, and Pollock]{Berk-2019a}
Graeme~D Berk, Andrew~JP Garner, Benjamin Yadin, Kavan Modi, and Felix~A
  Pollock.
\newblock Resource theories of multi-time processes: A window into quantum
  non-markovianity, 2019.
\newblock \eprint{arXiv:1907.07003}.

\bibitem[Liu and Yuan(2020)]{liu2020operational}
Yunchao Liu and Xiao Yuan.
\newblock Operational resource theory of quantum channels.
\newblock \emph{Physical Review Research}, 2\penalty0 (1):\penalty0 012035,
  2020.
\newblock \doi{10.1103/PhysRevResearch.2.012035}.

\bibitem[Liu and Winter(2019)]{Liu-2019b}
Zi-Wen Liu and Andreas Winter.
\newblock Resource theories of quantum channels and the universal role of
  resource erasure.
\newblock 2019.
\newblock \eprint{arXiv:1904.04201}.

\bibitem[Gour and Winter(2019)]{Gour-2019b}
Gilad Gour and Andreas Winter.
\newblock How to quantify a dynamical quantum resource.
\newblock \emph{Phys. Rev. Lett.}, 123:\penalty0 150401, Oct 2019.
\newblock \doi{10.1103/PhysRevLett.123.150401}.

\bibitem[Gour and Scandolo(2019)]{Gour2019ent}
Gilad Gour and Carlo~Maria Scandolo.
\newblock The entanglement of a bipartite channel.
\newblock 2019.
\newblock \eprint{arXiv:1907.02552}.

\bibitem[B\"{a}uml et~al.(2019)B\"{a}uml, Das, Wang, and Wilde]{Bauml-2019a}
Stefan B\"{a}uml, Siddhartha Das, Xin Wang, and Mark~M. Wilde.
\newblock Resource theory of entanglement for bipartite quantum channels.
\newblock 2019.
\newblock \eprint{arXiv:1907.04181}.

\bibitem[Jamio{\l}kowski(1972)]{jamiolkowski1972linear}
Andrzej Jamio{\l}kowski.
\newblock Linear transformations which preserve trace and positive
  semidefiniteness of operators.
\newblock \emph{Reports on Mathematical Physics}, 3\penalty0 (4):\penalty0
  275--278, 1972.
\newblock \doi{10.1016/0034-4877(72)90011-0}.

\bibitem[Choi(1975)]{choi1975completely}
Man-Duen Choi.
\newblock Completely positive linear maps on complex matrices.
\newblock \emph{Linear algebra and its applications}, 10\penalty0 (3):\penalty0
  285--290, 1975.
\newblock \doi{10.1016/0024-3795(75)90075-0}.

\bibitem[Vedral et~al.(1997)Vedral, Plenio, Rippin, and
  Knight]{vedral1997quantifying}
Vlatko Vedral, Martin~B Plenio, Michael~A Rippin, and Peter~L Knight.
\newblock Quantifying entanglement.
\newblock \emph{Physical Review Letters}, 78\penalty0 (12):\penalty0 2275,
  1997.
\newblock \doi{10.1103/PhysRevLett.78.2275}.

\bibitem[Barnum et~al.(1998)Barnum, Nielsen, and
  Schumacher]{Barnum1998transmission}
Howard Barnum, M.~A. Nielsen, and Benjamin Schumacher.
\newblock Information transmission through a noisy quantum channel.
\newblock \emph{Phys. Rev. A}, 57:\penalty0 4153--4175, Jun 1998.
\newblock \doi{10.1103/PhysRevA.57.4153}.

\bibitem[Gour(2019)]{Gour2019channel}
Gilad Gour.
\newblock Comparison of quantum channels by superchannels.
\newblock \emph{IEEE Transactions on Information Theory}, 65\penalty0
  (9):\penalty0 5880--5904, 2019.
\newblock \doi{10.1109/TIT.2019.2907989}.

\bibitem[Burniston et~al.(2020)Burniston, Grabowecky, Scandolo, Chiribella, and
  Gour]{Burniston2019measure}
John Burniston, Michael Grabowecky, Carlo~Maria Scandolo, Giulio Chiribella,
  and Gilad Gour.
\newblock Necessary and sufficient conditions on measurements of quantum
  channels.
\newblock \emph{Proceedings of the Royal Society A}, 476\penalty0
  (2236):\penalty0 20190832, 2020.
\newblock \doi{10.1098/rspa.2019.0832}.

\bibitem[Chiribella et~al.(2008{\natexlab{b}})Chiribella, D'Ariano, and
  Perinotti]{chiribella2008memory}
Giulio Chiribella, Giacomo~M D'Ariano, and Paolo Perinotti.
\newblock Memory effects in quantum channel discrimination.
\newblock \emph{Physical review letters}, 101\penalty0 (18):\penalty0 180501,
  2008{\natexlab{b}}.
\newblock \doi{10.1103/PhysRevLett.101.180501}.

\bibitem[Ziman(2008)]{Ziman2008PPOVM}
M\'ario Ziman.
\newblock Process positive-operator-valued measure: A mathematical framework
  for the description of process tomography experiments.
\newblock \emph{Phys. Rev. A}, 77:\penalty0 062112, Jun 2008.
\newblock \doi{10.1103/PhysRevA.77.062112}.

\bibitem[Gutoski and Watrous(2007)]{gutoski2007toward}
Gus Gutoski and John Watrous.
\newblock Toward a general theory of quantum games.
\newblock In \emph{Proceedings of the thirty-ninth annual ACM symposium on
  Theory of computing}, pages 565--574, 2007.
\newblock \doi{10.1145/1250790.1250873}.

\bibitem[Bennett et~al.(1999)Bennett, DiVincenzo, Fuchs, Mor, Rains, Shor,
  Smolin, and Wootters]{Bennett-1999a}
Charles~H. Bennett, David~P. DiVincenzo, Christopher~A. Fuchs, Tal Mor, Eric
  Rains, Peter~W. Shor, John~A. Smolin, and William~K. Wootters.
\newblock Quantum nonlocality without entanglement.
\newblock \emph{Phys. Rev. A}, 59:\penalty0 1070--1091, Feb 1999.
\newblock \doi{10.1103/PhysRevA.59.1070}.

\bibitem[Peres and Wootters(1991)]{Peres-1991a}
Asher Peres and William~K. Wootters.
\newblock Optimal detection of quantum information.
\newblock \emph{Phys. Rev. Lett.}, 66:\penalty0 1119--1122, Mar 1991.
\newblock \doi{10.1103/PhysRevLett.66.1119}.

\bibitem[Chitambar and Hsieh(2013)]{Chitambar-2013a}
Eric Chitambar and Min-Hsiu Hsieh.
\newblock Revisiting the optimal detection of quantum information.
\newblock \emph{Phys. Rev. A}, 88:\penalty0 020302, Aug 2013.
\newblock \doi{10.1103/PhysRevA.88.020302}.

\bibitem[Chiribella et~al.(2008{\natexlab{c}})Chiribella, D'Ariano, and
  Perinotti]{Chiribella-2008a}
G.~Chiribella, G.~M. D'Ariano, and P.~Perinotti.
\newblock Quantum circuit architecture.
\newblock \emph{Phys. Rev. Lett.}, 101:\penalty0 060401, Aug
  2008{\natexlab{c}}.
\newblock \doi{10.1103/PhysRevLett.101.060401}.

\bibitem[Chiribella et~al.(2009)Chiribella, D'Ariano, and
  Perinotti]{Chiribella-2009a}
Giulio Chiribella, Giacomo~Mauro D'Ariano, and Paolo Perinotti.
\newblock Theoretical framework for quantum networks.
\newblock \emph{Phys. Rev. A}, 80:\penalty0 022339, Aug 2009.
\newblock \doi{10.1103/PhysRevA.80.022339}.

\bibitem[Gutoski(2012)]{Gutoski2012distance}
Gus Gutoski.
\newblock On a measure of distance for quantum strategies.
\newblock \emph{Journal of Mathematical Physics}, 53\penalty0 (3):\penalty0
  032202, 2012.
\newblock \doi{10.1063/1.3693621}.

\bibitem[Gutoski(2009)]{Gutoski_PhD}
Gus Gutoski.
\newblock \emph{Quantum strategies and local operations}.
\newblock PhD thesis, University of Waterloo, 2009.
\newblock Available at arXiv.org e-Print quant-ph/1003.0038.

\bibitem[Chitambar et~al.(2014)Chitambar, Leung, Man{\v{c}}inska, Ozols, and
  Winter]{Chitambar2014}
Eric Chitambar, Debbie Leung, Laura Man{\v{c}}inska, Maris Ozols, and Andreas
  Winter.
\newblock Everything you always wanted to know about {LOCC} (but were afraid to
  ask).
\newblock \emph{Communications in Mathematical Physics}, 328\penalty0
  (1):\penalty0 303--326, May 2014.
\newblock \doi{10.1007/s00220-014-1953-9}.

\bibitem[D\"ur et~al.(2000)D\"ur, Vidal, and Cirac]{Dur2000}
W.~D\"ur, G.~Vidal, and J.~I. Cirac.
\newblock Three qubits can be entangled in two inequivalent ways.
\newblock \emph{Phys. Rev. A}, 62:\penalty0 062314, Nov 2000.
\newblock \doi{10.1103/PhysRevA.62.062314}.

\bibitem[Chitambar et~al.(2020)Chitambar, de~Vicente, Girard, and
  Gour]{chitambar2020entanglement}
Eric Chitambar, Julio~I de~Vicente, Mark~W Girard, and Gilad Gour.
\newblock Entanglement manipulation beyond local operations and classical
  communication.
\newblock \emph{Journal of Mathematical Physics}, 61\penalty0 (4):\penalty0
  042201, 2020.
\newblock \doi{10.1063/1.5124109}.

\bibitem[Christandl et~al.(2019)Christandl, M{\"u}ller-Hermes, and
  Wolf]{christandl2019composed}
Matthias Christandl, Alexander M{\"u}ller-Hermes, and Michael~M Wolf.
\newblock When do composed maps become entanglement breaking?
\newblock In \emph{Annales Henri Poincar{\'e}}, volume~20, pages 2295--2322.
  Springer, 2019.
\newblock \doi{10.1007/s00023-019-00774-7}.

\bibitem[Chru{\'{s}}ci{\'{n}}ski and Kossakowski(2006)]{Chruscinski2006}
Dariusz Chru{\'{s}}ci{\'{n}}ski and Andrzej Kossakowski.
\newblock On partially entanglement breaking channels.
\newblock \emph{Open Systems {\&} Information Dynamics}, 13\penalty0
  (1):\penalty0 17--26, Mar 2006.
\newblock ISSN 1573-1324.
\newblock \doi{10.1007/s11080-006-7264-7}.

\bibitem[Werner(1989)]{Werner1989}
Reinhard~F. Werner.
\newblock Quantum states with einstein-podolsky-rosen correlations admitting a
  hidden-variable model.
\newblock \emph{Phys. Rev. A}, 40:\penalty0 4277--4281, Oct 1989.
\newblock \doi{10.1103/PhysRevA.40.4277}.

\bibitem[Gurvits and Barnum(2002)]{Gurvits2002}
Leonid Gurvits and Howard Barnum.
\newblock Largest separable balls around the maximally mixed bipartite quantum
  state.
\newblock \emph{Phys. Rev. A}, 66:\penalty0 062311, Dec 2002.
\newblock \doi{10.1103/PhysRevA.66.062311}.

\bibitem[Christandl(2012)]{christandl2012ppt}
M~Christandl.
\newblock {PPT} square conjecture.
\newblock In \emph{Banff International Research Station workshop: Operator
  structures in quantum information theory}, 2012.

\bibitem[Lami and Giovannetti(2015)]{lami2015entanglement}
Ludovico Lami and Vittorio Giovannetti.
\newblock Entanglement--breaking indices.
\newblock \emph{Journal of Mathematical Physics}, 56\penalty0 (9):\penalty0
  092201, 2015.
\newblock \doi{10.1063/1.4931482}.

\bibitem[Lami and Giovannetti(2016)]{lami2016entanglement}
Ludovico Lami and Vittorio Giovannetti.
\newblock Entanglement-saving channels.
\newblock \emph{Journal of Mathematical Physics}, 57\penalty0 (3):\penalty0
  032201, 2016.
\newblock \doi{10.1063/1.4942495}.

\bibitem[Rahaman et~al.(2018)Rahaman, Jaques, and
  Paulsen]{rahaman2018eventually}
Mizanur Rahaman, Samuel Jaques, and Vern~I Paulsen.
\newblock Eventually entanglement breaking maps.
\newblock \emph{Journal of Mathematical Physics}, 59\penalty0 (6):\penalty0
  062201, 2018.
\newblock \doi{10.1063/1.5024385}.

\bibitem[Kennedy et~al.(2018)Kennedy, Manor, and
  Paulsen]{10.5555/3370201.3370205}
Matthew Kennedy, Nicholas~A. Manor, and Vern~I. Paulsen.
\newblock Composition of ppt maps.
\newblock \emph{Quantum Info. Comput.}, 18\penalty0 (5-6):\penalty0 472--480,
  2018.

\bibitem[Rosset et~al.(2018)Rosset, Buscemi, and Liang]{Rosset2018memory}
Denis Rosset, Francesco Buscemi, and Yeong-Cherng Liang.
\newblock Resource theory of quantum memories and their faithful verification
  with minimal assumptions.
\newblock \emph{Phys. Rev. X}, 8:\penalty0 021033, May 2018.
\newblock \doi{10.1103/PhysRevX.8.021033}.

\bibitem[Yuan et~al.(2019)Yuan, Liu, Zhao, Regula, Thompson, and
  Gu]{Xiao2019memory}
Xiao Yuan, Yunchao Liu, Qi~Zhao, Bartosz Regula, Jayne Thompson, and Mile Gu.
\newblock Universal and operational benchmarking of quantum memories, 2019.
\newblock \eprint{arXiv:1907.02521}.

\end{thebibliography}

\end{document}